\documentclass[a4paper, 11pt]{scrartcl}

\usepackage[utf8]{inputenc}
\usepackage[T1]{fontenc}
\usepackage{url} 

\usepackage{rotating}
\usepackage{amsmath,amssymb,mathtools}

\usepackage{amsthm}
\usepackage[dvipsnames]{xcolor}
\usepackage{graphicx}
\usepackage{subcaption}
\usepackage{enumitem}
\usepackage{nicefrac}
\usepackage{thm-restate}
\usepackage[colorlinks]{hyperref}
\hypersetup{
	allcolors=cyan,
	citecolor=blue,
	urlcolor=blue
}
\usepackage{framed}
\usepackage{xspace}
\usepackage{lineno} 
\usepackage{comment}
\usepackage{stmaryrd}

\newcommand{\cref}{\autoref}
\newcommand{\Cref}{\autoref}
\newtheorem{theorem}{Theorem}
\newtheorem*{claim*}{Claim}
\newtheorem{lemma}[theorem]{Lemma}

\newtheorem{proposition}[theorem]{Proposition}
\newtheorem{q}{Open Problem}

\newtheorem*{conjecture}{Conjecture}

\newcommand{\leqlin}{\leq_{\text{lin}}}
\newcommand{\linearExtension}{linear extension\xspace}

\def\A{\ensuremath{\mathcal{A}}\xspace}

\def\R{\ensuremath{\mathbb{R}}\xspace}
\def\Z{\ensuremath{\mathbb{Z}}\xspace}
\def\Q{\ensuremath{\mathbb{Q}}\xspace}

\def\N{\ensuremath{\mathbb{N}}\xspace}
\def\P{\ensuremath{\mathcal{P}}\xspace}

\def\G{\ensuremath{\mathcal{G}}\xspace}

\def\ETR{\textsc{ETR}\xspace}

\def\ETRAMI{\ensuremath{\textsc{ETR-ami}}\xspace}
\def\UETRAMI{\ensuremath{\textsc{UETR-ami}}\xspace}
\def\NP{\textsc{NP}\xspace}
\def\ETRNew{\textsc{ETR-inv}\xspace}

\def\PETRNew{\textsc{Pla\-nar-ETR-inv}\xspace}
\def\UETR{\textsc{UETR}\xspace}
\def\UETRNew{\textsc{Constrained-UETR}\xspace}

\def\ER{\ensuremath{\exists \mathbb{R}}\xspace}
\def\FER{\ensuremath{\forall \exists \mathbb{R}}\xspace}
\def\EFR{\ensuremath{ \exists \forall \mathbb{R}}\xspace}

\def\AU{\textsc{Area Uni\-ver\-sa\-li\-ty}$_{\geq 0}$\xspace}
\def\AUpos{\textsc{Area Uni\-ver\-sa\-li\-ty}\xspace}
\def\AGP{\textsc{Art Gallery Problem}\xspace}

\def\PPA{\textsc{Prescribed Area Extension}\xspace}%Planar Prescribed Area*
\def\PA{\textsc{Area Universality for Tri\-ples wPFA}\xspace}
\def\PAA{\textsc{Prescribed Area}\xspace}

\newcommand{\const}[1]{\ensuremath{\left\llbracket \, #1 \,\right\rrbracket}\xspace}

\graphicspath{{figures/}}

\newcommand*\patchAmsMathEnvironmentForLineno[1]{%
	\expandafter\let\csname old#1\expandafter\endcsname\csname #1\endcsname
	\expandafter\let\csname oldend#1\expandafter\endcsname\csname end#1\endcsname
	\renewenvironment{#1}%
	{\linenomath\csname old#1\endcsname}%
	{\csname oldend#1\endcsname\endlinenomath}}% 
\newcommand*\patchBothAmsMathEnvironmentsForLineno[1]{%
	\patchAmsMathEnvironmentForLineno{#1}%
	\patchAmsMathEnvironmentForLineno{#1*}}%
\AtBeginDocument{%
	\patchBothAmsMathEnvironmentsForLineno{equation}%
	\patchBothAmsMathEnvironmentsForLineno{align}%
	\patchBothAmsMathEnvironmentsForLineno{flalign}%
	\patchBothAmsMathEnvironmentsForLineno{alignat}%
	\patchBothAmsMathEnvironmentsForLineno{gather}%
	\patchBothAmsMathEnvironmentsForLineno{multline}%
}
\usepackage[absolute]{textpos}
\begin{document}

\title{Completeness for the Complexity Class $\forall \exists \mathbb{R}$ and Area-Universality
  \thanks{A preliminary version appears in the proceedings of WG 2018 \cite{kleist2}.
  Moreover, a video presenting this paper is available at \url{https://youtu.be/OQkACiNS66o}.
  Tillmann Miltzow was generously supported by the Netherlands Organisation for Scientific Research (NWO) under project no. 016.Veni.192.250. 
Pawe{\l} Rz\k{a}\.zewski was supported by the ERC Starting Grant CUTACOMBS (grant agreement No.~714704).
  }  
}

\let\email=\url
\def\SC{\scshape}
\let\bf=\bfseries

\author{
	\parbox{6.7cm}{\center
		{\SC Michael Gene Dobbins}\\[3pt]
		\small
		\email{mdobbins@binghamton.edu}\\
		{Binghamton University}}
	\and\hspace{-2cm}
	\parbox{6.7cm}{\center
			{\SC Linda Kleist}\\[3pt]
	\small
	\email{kleist@ibr.cs.tu-bs.de}\\
	{Technische Universit\"at Braunschweig}}
	\and\hspace{-0.6cm}
	\parbox{6.7cm}{\center
		{\SC Tillmann Miltzow}\\[3pt]
		\small
		\email{t.miltzow@uu.nl}\\
		{Utrecht University}}
	\and\hspace{-2cm}
	\parbox{6.7cm}{\center
		{\SC Pawe{\l} Rz\k{a}\.zewski}\\[3pt]
		\small
		\email{p.rzazewski@mini.pw.edu.pl}\\
		{Warsaw University of Technology \\ and University of Warsaw}}
}

\maketitle
\vspace{-12pt}

\begin{abstract}
Exhibiting a deep connection between purely geometric problems and real algebra, the complexity class \ER plays a  crucial role in the study of geometric problems.
Sometimes \ER is referred to as the `real analog' of~\NP. 
While \NP is a class of computational problems that deals with existentially quantified \emph{boolean} variables, \ER deals with existentially quantified \emph{real} variables.

In analogy to $\Pi_2^p$ and $\Sigma_2^p$ in the famous polynomial hierarchy, we study the  complexity classes \FER and \EFR with \emph{real} variables.
Our main interest is  the \AUpos problem, where we are given a plane graph~$G$, and ask if for each assignment of areas to the inner faces of $G$, there exists a straight-line drawing of $G$ realizing the assigned areas.
We conjecture that \AUpos is \FER-complete and support this conjecture  by  proving \ER- and \FER-completeness 
of two variants of \AUpos.
To this end, we introduce tools to prove \FER-hardness and membership.
Finally, we present geometric problems as candidates 
for \FER-complete problems. These problems have connections 
to the concepts of imprecision, robustness, and extendability.  
%\keywords{$\forall \exists \mathbb{R}$ \and Existential Theory of the Reals \and Area Universality}
\end{abstract}

\noindent\textbf{Keywords:} complexity class, existential theory of the reals, universal existential theory of the reals, planar graph, face area, area-universality

\begin{textblock}{20}(0, 11.5)
\includegraphics[width=40px]{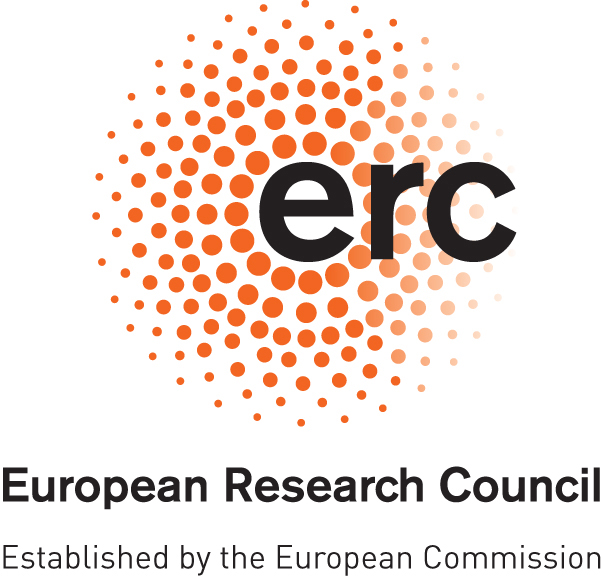}%
\end{textblock}
\begin{textblock}{20}(-0.25, 11.9)
\includegraphics[width=60px]{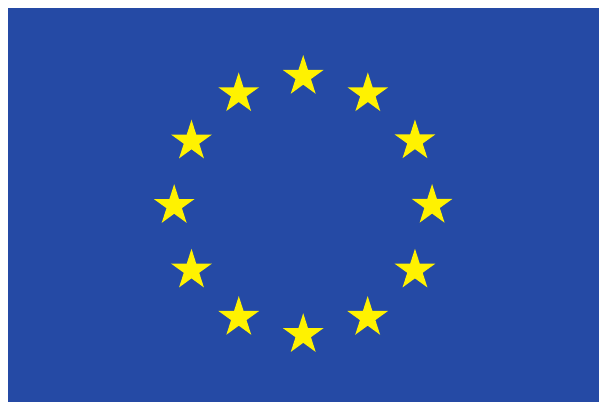}%
\end{textblock}

% !TeX root = main-algorithmica.tex
\section{Introduction}
Size is a very intuitive visual variable. Therefore, statistics are frequently illustrated by distorted maps where regions are scaled according to the population, number of births, average income or some other parameter of interest. Such distorted maps, known as \emph{cartograms}, inspire the investigation of  problems related to face areas 
in straight-line drawings of planar graphs.

A plane graph is a graph together with a \emph{planar drawing}, in which edges may intersect only in common vertices.
Two  planar drawings of a graph are \emph{equivalent} if they have the same outer 
face and the same set of inner faces; i.e., walking on the face boundaries (in counter clockwise direction) yields the  same (directed) cycles.
Let $G$ be a plane graph and let $F$ be the set of inner faces of $G$. 
An {\em area assignment} is a real-valued function $\A \colon F \to \R^+$. 
We say that a drawing $G'$ is {\em \A-realizing}, if 
$G'$ is a straight-line drawing equivalent to $G$ in which 
the area of each $f \in F$ is exactly~$\A(f)$. 
If $G$ has an $\A$-realizing drawing, we say that $\A$ is 
{\em realizable}.
A plane graph $G$ is {\em area-universal} if every  area assignment is realizable.
In this paper, we study the computational problem of deciding whether a given plane graph is area-universal. 
We denote this problem by \AUpos.

\begin{framed}
	\noindent{\bf \AUpos}\\
	\textbf{Input:} A  plane graph $G=(V,E)$.\\    
	\textbf{Question:}  Does every positive area  assignment $\A$ of $G$ have a planar $\A$-realizing drawing?
\end{framed}

With cartograms in mind, it is also natural to allow face areas of 0. 
Note that for a planar realizing drawing, the area assignment must be \emph{positive}, i.e., all assigned areas are positive.
In order to be able to realize face areas of 0, we slightly relax the planarity condition for the set of realizing drawings: We consider a straight-line drawing as a \emph{crossing-free} drawing of a plane graph $G$ if 
for  every $ \varepsilon >0$,  the  vertex coordinates  can  be  perturbed  within  a  ball  of  radius $\varepsilon$ such that the result is a planar drawing equivalent to $G$.
In other words, we extend the set of planar drawings by the (non-planar) \emph{degenerate} drawings that can be obtained as the limit of a sequence of planar drawings.
 In the context of weakly simple polygons, this notion is also known under the term \emph{weak embedding}; see also \cite{weekEmbeddings}.
Because all considered drawings are crossing-free and  straight-line, we simply call them drawings from now on; we use the terms degenerate or non-degenerate to distinguish between planar and crossing-free drawings if necessary. 
This results in the following closely related decision problem.

\begin{framed}
\noindent{\bf \AU}\\
\textbf{Input:} A  plane graph $G=(V,E)$.\\ 
\textbf{Question:}  Does every (non-negative) area  assignment $\A$ of $G$ have a crossing-free $\A$-realizing drawing?
\end{framed}

Interestingly, for a large family of plane graphs, the problems \AUpos and \AU coincide:
A compactness argument shows that for the area-universality of a triangulation (even more generally, for every plane graph with a triangular outer face) it makes in fact no difference whether or not faces of 0 area are allowed, i.e., all positive area assignments of a triangulation are realizable if and only if all non-negative area assignments are realizable. For more details, we refer to \cite[Corollary 5]{kleist1Journal} and \cite[Corollary 3]{kleistPhD}.
As we will also point out in \cref{sec:stateOfTheArt}, allowing face areas of 0 turned out to be a useful tool to significantly simplify proofs and enable (more) elegant proofs, e.g., in the context of disproving area-universality \cite{kleist1Journal}.

\subsection{Introduction to the Complexity Class \FER}

%Existential Theory of the Reals
When investigating geometric problems, one often discovers that their instances can be described by  a system of polynomial equations and inequalities $\Phi$, so that real-valued variable assignments that satisfy $\Phi$ correspond to solutions of the original geometric problem.
The variables encode the configuration of the geometric objects and the quantifier-free formula~$\Phi$ describes the relations between them. 
\textsc{Existential Theory of the Reals} (\ETR) is a computational problem that takes a first-order formula containing only existential quantifiers: $\exists X=(X_1,X_2,\ldots,X_n) \colon \Phi(X)$ where $\Phi$ may contain the symbols \[0, 1, +, \cdot , =, <, \land, \neg, (, ), X_1, \dots,X_n\] and asks whether it is true or not over the reals. 
 An example of such an instance is the following formula\footnote{We use the notation of Matou\v{s}ek~\cite{matousekSegments} where capital letters denote the names of variables, while lower-case letters denote the values of variables (numbers or logic values).}: 
$\exists (X_1, X_2)  : (X_1^2 + X_2^2 > 1) \, \land \, (3X_1 + 2X_2 = 10).$

%\ER
The complexity class \ER consists of all decision problems that are many-one reducible to \ETR by a Turing machine in a polynomial number of steps.
Many natural geometric problems appear to be \ER-complete, i.e.,\ they are in \ER~and \ETR is reducible to them.
Prominent examples include separability of pseudoline arrangements~\cite{matousekSegments,Mnev,DBLP:journals/mst/SchaeferS17},  recognizing unit disk graphs~\cite{DBLP:journals/dcg/KangM12} and intersection graphs of segments~\cite{matousekSegments,DBLP:conf/gd/Schaefer09},  the art gallery problem~\cite{AbrahamsenAM18STOC}, geometric packing~\cite{PackETR}, or training neural networks~\cite{NeuronETR}.

Despite tremendous work in real algebraic geometry, we  do not know any simple algorithm to decide \ETR.
Consequently, the \ER-completeness of  any of the problems mentioned above implies that there is little hope that they could be solved using some simple algorithms, and most known algorithms have some algebraic flavor.
In fact, we usually need to employ certain algebraic tools even to show that these problems are \emph{decidable}.
In conclusion, the \ER-completeness of many geometric problems reflects deep algebraic connections between these problems and real algebra. 

%UETR and \FER
While geometric problems that are \ER-complete usually ask for
the  existence of certain objects, satisfying some semi-algebraic 
properties, the nature of area-universality seems to be different. We therefore 
define the  complexity class \FER as the set of all problems that reduce 
in polynomial time to \textsc{Universal Existential Theory of the Reals} (\UETR).
The input of \UETR is a first order formula over the reals in prenex form, which starts with a block of universal quantifiers followed by a block of existential quantifiers and is otherwise quantifier-free. 
We ask if the formula is true.
An example of such an instance is: 
 \[ \forall Y_1,Y_2 \,  \, \exists X_1  : (Y_1^2 + X_1^2 \geq 1) \, \land \, (3Y_1 + 2Y_2 = 10X_1).\]
\begin{framed}
\noindent{\bf{\textsc{Universal Existential Theory of the Reals} (\UETR)}}\\
\textbf{Input:} A formula $\left(\forall \, Y=(Y_1,Y_2,\ldots,Y_m)\right) \left(\exists \, X=(X_1,X_2,\ldots,X_n)\right) \colon \Phi(X,Y)$ over the reals
where $\Phi(X,Y)$ is quantifier-free.\\
\textbf{Question:} Is the input formula true?
\end{framed}

\paragraph{Relation of the complexity classes.} The complexity class \EFR is defined analogously. Clearly, \ER is contained in both, \EFR and \FER. It is easy to observe and well-known that \NP is contained in \ER. Highly non-trivial is the containment of \EFR and \FER in PSPACE, which follows from a more general result that deciding first-order formulae over the reals with bounded number 
of quantifier blocks is in PSPACE~\cite{basu2006algorithms}.
For all we know, all these complexity classes could collapse, as it is open whether \NP and PSPACE constitute two different or the same complexity class, see \cref{fig:ComplexityZoo}.
However, \ER $\neq$ \FER can be believed with similar confidence as NP\,$\neq \Pi^p_2$.
In addition, it is known that the algebraic expressibility of \FER-formulae is larger than \ER-formulae~\cite{davenport1988real}.

\begin{figure}[b!]
\centering
\includegraphics[page=8]{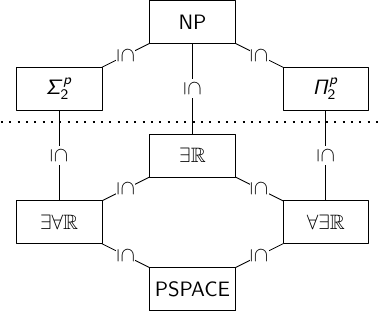}
\caption{Containment diagram of the complexity classes.}
\label{fig:ComplexityZoo}
\end{figure}

It is worth mentioning that Blum et.\ al.~\cite{blum} also introduced 
a hierarchy of complexity classes analogous to the complexity class NP, but over the reals (or other rings). 
Their canonical model of computation is the so-called Blum-Shub-Smale machine (BSS). 
The main difference between these approaches is that BSS accepts real numbers as input and the arithmetic operations over the reals can be performed at unit cost, while the classes discussed in our paper (\ER, \FER, \EFR) work with ordinary Turing machines, accepting only finite strings over finite alphabets.

\subsection{Our Results}
First of all, we show that both variants of area-universality belong to \FER.

\begin{restatable}{proposition}{AreaUniInFER}\label{pro:Containment}
\AUpos and \AU are contained in~\FER.
\end{restatable}
The idea of the proof (presented in Section \ref{sec:auinfer})  is to use a block of universal quantifiers to describe the area assignment and the block of existential quantifiers to describe the placement of the vertices of the drawing of~$G$.
While it is straightforward to show containement for \AUpos, the challenge for proving \cref{pro:Containment} for \AU lies in the fact that we allow for degenerate realizing drawings. \cref{thm:CookLevin} enables us to exploit the fact that testing whether a drawing is crossing-free lies in $P$. Interestingly, this is far more intricate than deciding whether a drawing is planar.

Recently, Erickson, van der Hoog and Miltzow showed a result analogous 
to the Cook Levin Theorem for \ER-membership~\cite{RobustComputation}.
It is straight-forward to generalize their proof to \FER.
The gist of the theorem is that we can can
use real RAM algorithms to establish \FER-membership.
(See \cref{sec:Cook-Levin} for precise definitions.)
\begin{restatable}{theorem}{CookLevin} \label{thm:CookLevin}
	For every discrete decision problem $Q$,
	there exists a real challenge-\-veri\-fi\-ca\-tion algorithm
	if and only if $Q\in \FER$.
\end{restatable}

In view of \cref{pro:Containment}, we believe that a stronger statement holds. 
\begin{conjecture}
\AUpos is \FER-complete.  
\end{conjecture}
While this conjecture, if true, would show that \AUpos is a difficult problem in an algebraic and combinatorial sense, it would also give  the first known natural geometric problem that is complete for \FER.

Unfortunately, in contrast to most other algorithmic problems,
studying area-univer\-sal\-ity seems very difficult even on small instances.
This is particularly unfortunate since most hardness results are based
on constructing small gadgets that are  well understood and then combined in a very controlled and clever way.
As discussed in more detail in \cref{sec:stateOfTheArt}, there exist graphs on nine vertices for which we do not yet know whether or not they
are area-universal. This is due to (1) the continuous 
nature of possible area assignments 
and the vertex placements,  
 and (2) the doubly quantified interplay between the two.
Due to those difficulties, we focus on restricted variants
that give us some control over parts
of the final drawing. 
Although we are still some steps away from resolving our
conjecture, we believe that our results bring
us a big step forward.
Even more, many of the presented results and tools
are interesting by themselves
especially in the context 
of existing \ER-hardness results.
We consider two variants 
of \AUpos, each approaching the conjecture from 
a different direction.

Firstly, a natural relaxation is to 
drop the planarity restriction.
For a plane graph $G$ with vertex set $V$, the 
{\em face hypergraph} of $G$ has vertex set $V$, 
and its edges correspond to sets of vertices forming 
the faces (such hypergraphs were studied e.g. by Dvo\v{r}\'{a}k et al.~\cite{Dvorak}). Recall that a \emph{triangulation} is a maximal planar graph, i.e., every face is incident to three edges. Therefore, the 
face hypergraph of a plane triangulation is $3$-uniform, i.e.,\ each hyperedge has $3$ vertices. 
Clearly, \AUpos can be equivalently formulated in 
the language of face hypergraphs. This relation motivates 
the following  version of the problem in which we are fixing some of the areas, i.e, we consider instances \emph{with a partial fixed assignment (wPFA)}.

%-------------
\begin{framed}
\noindent{\bf \PA}\\
\noindent{\bf Input:}     
A set $V$ of vertices, a collection of vertex-triples 
$F\subseteq{V\choose 3}$, and a partial area assignment 
$\A' \colon F' \to \R^+$ for some $F' \subseteq F$.\\
\noindent{\bf Question:} For every 
$\A \colon F \to \R^+$ 
with $\A(f) = \A'(f)$ 
for all $f \in F'$, does there exist a placement of $V$ in 
the plane, such that the triangle of each $f \in F$ has area  $\A(f)$?
\end{framed}

\noindent
Note that \PA differs from \AUpos in two aspects: Firstly, we consider arbitrary triples in contrast to interior disjoint faces, and secondly, some areas are fixed. However, we are able to restrict to a linear number of triples and show the following hardness result.

\begin{restatable}{theorem}{triples}\label{thm:hyper}
 \PA is \FER-complete, even if the number of triples is linear in the number of vertices.
\end{restatable}
%-------------
For the proof of \cref{thm:hyper} we use gadgets similar 
to the \emph{von Staudt constructions} used to show 
the \ER-hardness of order-types (we refer the reader to the paper of Matou\v{s}ek~\cite{matousekSegments} for more information).
The proof can be found in \cref{sec:PA}.
\medskip
\medskip

Our second result concerns a variant, 
where we investigate the complexity of realizing 
a specific area assignment.
 \PAA denotes the following problem: 
 \begin{framed}
 	\noindent{\bf \PAA}\\
 	{\bf Input:}     A plane graph $G=(V,E)$ with an area assignment $\mathcal A$.\\
 	{\bf Question:}  Does there exist an \A-realizing drawing of $G$?
 \end{framed}
 
We believe that \PAA is \ER-complete. However, we are only able to show hardness of an \emph{extension} version of the problem, 
where some vertex positions are fixed and we seek a
placement of the remaining vertices realizing the prescribed areas. We call this problem \PPA.

\begin{framed}
\noindent{\bf \PPA}\\
{\bf Input:}     A plane graph $G=(V,E)$, a (positive) area assignment $\mathcal A$, fixed positions for $V'\subseteq V$.\\
{\bf Question:}  Does there exist a (planar)  \A-realizing drawing of $G$ respecting the positions of all vertices in $V'$?
\end{framed}

We show the following hardness result.
\begin{restatable}{theorem}{PPAthm}\label{thm:PPA}
	\PPA is \ER-complete. 
\end{restatable}

The proof of \autoref{thm:PPA} can be found in~\cref{sec:PPA}.
Let us point out that, along with the recent result of Lubiw et al.~\cite{LubiwDrawingExtension}, 
\autoref{thm:PPA} is one of the \emph{first} \ER-hardness results concerning drawings of planar graphs in the plane. 
\medskip

%-------------
To conclude and to motivate further research, we present problems that are interesting candidates for \FER-complete problems in~\cref{sec:Others}. These problems are linked to notions such as  robustness, imprecision, and extendability. 
To the best of our knowledge, no further problems are
known to be \FER-complete.

\subsection{Area-Universality -- State of the Art \& Related Work}\label{sec:stateOfTheArt}
In this section, we present the state of the art of area-universality. We present some ideas and challenges of how to prove and disprove area-universality.

\paragraph{Area-Universal Graph Families.}
Despite the fact that area-universality seems to be a strong property, it is straightforward to observe that it holds for \emph{stacked triangulations}, also known as \emph{plane 3-trees} or \emph{Apollonian networks}.
A stacked triangulation $T$ is defined recursively by subdividing a triangle $t$ of a stacked triangulation~$T'$ into three smaller triangles.
An area assignment of $T$ can be realized by first realizing~$T'$ so that $t$ has the total area of the three smaller triangles, and then subdividing~$t$ accordingly.
Moreover, it is easy to see that if a graph is area-universal, then each of its subgraphs is also area-universal. 
Thus, subgraphs of stacked triangulations are area-universal. Biedl and Vel\'azquez~\cite{biedl} additionally studied the grid size of realizing drawings of subgraphs of stacked triangulations.

The first graphs that were shown to be area-universal were plane cubic graphs, as proven already in 1992 by Thomassen~\cite{thomassen}. However, unlike for stacked triangulations, this inductive argument is highly technical. The proof makes extensive use of the fact that if a vertex $v$ of degree 3 is placed collinearly with two of its neighbors $u$ and $w$, then $v$ may be freely moved on the segment between $u$ and $w$. This fact allows to lay out a path in  a straight-line fashion in a step of the induction. Alternatively, an edge is contracted. Consequently, the resulting realizing drawings are highly degenerate.

More recently, Kleist~\cite{kleist1,kleist1Journal}  showed that the 1-subdivision of every plane graph is area-universal. In other words, every area assignment of a plane graph is realizable if the straight-line property is relaxed so that each edge may have  one bend. This result builds on the fact that every rectangular layout (dissection of a rectangle into rectangles) has a weakly-equivalent rectangular layout realizing prescribed areas~\cite{Eppstein,felsner,Floorplans}. In this setting, positive area assignments are guaranteed to have non-degenerate (planar) realizing drawings.

For straight-line drawings, however, deciding if a graph is area-universal seems to be a challenging problem even for very small triangulations. [The interested reader may play with 4-connected triangulations on seven or eight vertices or, alternatively, with the graph on nine vertices in \cref{fig:triang9}, see our discussion later.] 
Kleist~\cite{kleist3} developed  a method to prove the area-universality of triangulations  with special vertex orderings. The method is based on a sufficient criterion for area-universality that only requires the investigation of one area assignment. Among others, the criterion uses the fact that for a triangulation, the realizability of an area assignment can be characterized by a real solution of an equation system that describes the area of each triangular face with a determinant equation. The machinery can be used to characterize the area-universality of a graph family, called accordion graphs: An accordion graph $\mathcal K_\ell$ can be obtained from the plane octahedron graph by subdividing an edge of the central triangle by~$\ell$ vertices and introducing $2\ell$ additional edges such that the new vertices are of degree 4. For an example consider \cref{fig:acc}. 
Obviously, accordion graphs are structurally very similar. However, in terms of area-universality there is a surprising distinction: An accordion graph is area-universal if and only if $\ell$ is odd. This shows that there exists a very fine line between graphs that are and those that are not area-universal. In particular, this exhibits the sensitivity of area-universality to small
local changes. Thus, area-universality might be a \emph{global} property.

%------------------
\begin{figure}[b]
	\centering
	\begin{subfigure}[t]{.24\textwidth}
		\centering
		\includegraphics[page=2]{triang9}
		\caption{An accordion graph with $\ell=3$.}
		\label{fig:acc}
	\end{subfigure}
	\hfill
	\begin{subfigure}[t]{.7\textwidth}
		\centering
		\includegraphics{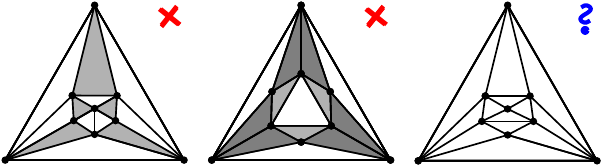}
		\caption{These three plane triangulations are drawings of the same abstract planar graph. While the left two are not area-universal, it is open whether the right one is.}
		\label{fig:triang9}
	\end{subfigure}
	\caption{Illustration of interesting triangulations.}
\end{figure}
%------------------
%

As a further interesting candidate, it is conjectured that plane bipartite graphs are area-universal. Evans et al.~\cite{kleistQuad,kleistPhD} developed the first tools to tackle this problem, and used them to confirm the conjecture for large subclasses of plane bipartite graphs, as well as for small plane bipartite graphs on up to 13 vertices. 

\paragraph{Disproving Area-Universality.}
In terms of negative results, it was already known to Rin\-gel~\cite{ringel1990equiareal} in 1990 that the octahedron graph is not area-universal. For a long time, the octahedron and its supergraphs remained the only known non-area-universal graphs.
Kleist~\cite{kleist1,kleist1Journal} introduced the first 
infinite family of non-area-universal graphs, that is not defined 
as the supergraphs of a single non-area-universal graph. 
More precisely, she presented a simple combinatorial argument to show that all Eulerian triangulations (different from the triangle) are not area-universal.  Note that the octahedron is also an Eulerian triangulation. Among other things, the elegance of the argument relies on the fact of using face areas of 0. Recall that  for the area-universality of a triangulation, it makes no difference whether or not faces of 0 area are allowed.

Moreover, the above results have several interesting consequences. 
For instance, there exist plane graphs and area assignments such that no drawing 
approximates the area assignment by a constant factor $c$, i.e., there
is no drawing such that
for every inner face~$f$ with assigned 
area $A$ it holds that $\nicefrac{1}{c}\cdot A \leq \textsc{area}(f)\leq c \cdot A$.
Combined with geometric arguments, the developed technique can be used to disprove area-universality of other graphs such as the icosahedron graph~\cite{kleist1Journal} or other small triangulations~\cite{Henning,kleistPhD}.
The fact that the icosahedron is not area-universal shows that high connectivity of a graph does not imply area-universality. Furthermore, area-universality is not a minor-closed property, as every grid graph is area-universal~\cite{kleistQuad,tableCartograms,kleistPhD}, but the octahedron is not area-universal, although it is a minor of the grid.

\paragraph{Understanding Small Graphs.}
As mentioned before, understanding properties of realizing drawings of small graphs can be very useful;  in particular in order to construct gadgets for hardness proofs.
However, deciding whether a graph is area-universal is already a challenge for quite small triangulations.
Using the tools developed by Kleist~\cite{kleist1Journal,kleist3,kleistPhD} one can characterize the area-universality of triangulations with special vertex orderings on up to ten vertices. In fact, it suffices to study 4-connected triangulations.  
The most interesting and smallest graph with unknown status is the triangulation depicted in the right of \cref{fig:triang9}; this triangulation does not have one of the special vertex orderings mentioned above.
The underlying planar graph has three distinct embeddings which are depicted in \cref{fig:triang9}. While the leftmost plane graph has a non-realizable area assignment consisting of 0's and 1's (called 01-assignments for simplicity), all 01-assignments of the middle graph are realizable. Nevertheless, the middle graph is not area-universal. 
If the rightmost triangulation was area-universal, then this would show that area-universality is a property of plane rather than planar graphs.

\section{Preliminaries}
\label{sec:Prep}
In this section, we firstly introduce restricted but hard variants of \ETR and \UETR. These variants will be the base problems for our reductions. Secondly, we present a Cook-Levin analog which we then use in order to prove the containment of \AUpos and \AU  in \FER. 

\subsection{Toolbox: Hard variants of \ETR and \UETR}
\label{sec:toolbox}
Abrahamsen et al.\ showed that the 
following problem is \ER-complete~\cite{AbrahamsenAM18STOC}.

\begin{framed}
\noindent{\textbf \ETRNew} \\
\textbf{Input:} A formula over the reals of the form
$\left (\exists \; X_1,X_2,\ldots,X_n \right ) \colon\Phi,$
where $\Phi$ is a conjunction of constraints of the following form: $X=1$ (introducing a constant 1), $X+Y=Z$ (addition), $X\cdot Y = 1$ (inversion), with $X,Y,Z \in \{X_1,\ldots,X_n\}$.  Additionally, $\Phi$ is either unsatisfiable or has a solution where each variable is within~$[1/2,2]$.\\
 {\bf Question:}   Is the input formula true?  
 \end{framed}
 
In order to define some even more restricted variant of \ETRNew, we need one more definition.
Consider a formula $\Phi$ of the form $\Phi = \Phi_1 \land \Phi_2 
\land \ldots \land \Phi_m$, where each $\Phi_i$ is a quantifier-free 
formula of the 
first-order theory of the reals with variables $X_1,X_2,\ldots,X_n$, which 
uses arithmetic operators and comparisons ($=,<,\leq$) but no logic symbols.
The {\em incidence graph} of $\Phi$ is the bipartite graph with vertex set $\{X_1,X_2,\ldots,X_n\} \cup \{\Phi_1,\Phi_2,\ldots,\Phi_m\}$ that has an edge $X_i\Phi_j$  if and only if the variable $X_i$ appears in the subformula $\Phi_j$.
By \PETRNew we denote the variant of \ETRNew where the incidence graph of $\Phi$ is planar and $\Phi$ is either unsatisfiable or has a solution with all variables within $[1/2,4]$. 

Note that Lubiw, Mondal and Miltzow introduced another version of \PETRNew, which has inequality instead of equality constraints~\cite{LubiwDrawingExtension}.

\begin{theorem} \label{petr}
\PETRNew is \ER-complete.
\end{theorem} 
\begin{proof}
Consider an instance 
$\left (\exists \; X_1,X_2,\ldots,X_n \right )\colon \Phi$ of \ETRNew. Let $G$ be some embedding of $G(\Phi)$ in $\R^2$. Suppose that $G$ is not crossing-free and consider a pair of crossing edges. Let $X$ and $Y$ denote the variables corresponding to (one endpoint of) these edges as in~\cref{fig:crossing}. 

\begin{figure}[htb]
	\centering   \includegraphics{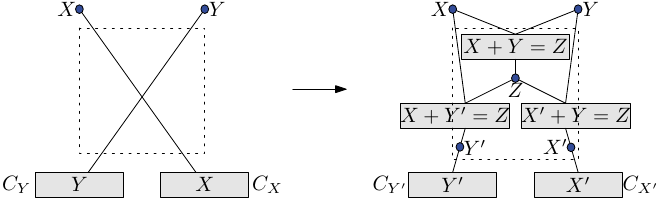}
	\caption{Eliminating crossings.}   \label{fig:crossing}
\end{figure}

We introduce three new existential variables $X',Y',Z$ and three constraints: 
 $X + Y =  Z$, $ X + Y' = Z$, and $X' + Y = Z$.
  Observe that these constraints ensure that $X = X'$ and $Y=Y'$.
 Moreover, the  embedding of $G$ can be modified so that the new 
 incidence graph $G'$ has strictly fewer crossings:
$G'$ loses the considered crossing and no new crossing is introduced.  We repeat this procedure until the incidence 
graph of the obtained formula is planar. Finally, note that 
$ 1 \leq Z = X + Y \leq 2 + 2=4 $ whenever $1/2 \leq X,Y \leq 2$, and the number of new variables and constraints is polynomial in $|\Phi|$, since the number of variables in 
 each constraint in \ETRNew is at most three. 
\end{proof}

Now we introduce a restricted variant of \UETR.

\begin{framed} 
\noindent{\textbf \UETRNew} \\
    \textbf{Input:} 
A formula  $\left (\forall \; Y_1,\ldots,Y_m \in \R^+ \right ) \left (\exists \; X_1,\ldots,X_n  \in \R^+ \right ) \colon \Phi(X,Y)$ over the reals, 
where  $\Phi$ is a conjunction of constraints of 
the form: $X=1$ (introducing constant 1), $X+Y=Z$ 
(addition), $X\cdot Y = Z$ (multiplication), 
with $X,Y,Z \in \{X_1,\ldots,X_n,Y_1,\ldots,Y_m\}$. \\
  \textbf{Question:} Is the input formula true?     
 \end{framed}

\UETRNew can be seen as a variant of \FER that is simplified in a way analogous to a \ER-complete  
variant of \ETR called \textsc{Ineq} \cite{matousekSegments,DBLP:journals/mst/SchaeferS17}. 
Similarly, we will show that \UETRNew is \FER-complete. 
\begin{restatable}{theorem}{UTRNEWTHM}\label{thm:UETRNEW}
\UETRNew is \FER-complete.
\end{restatable}
The proof of this theorem relies on the toolbox by Abrahamsen and Miltzow~\cite{DynamicToolboxETRINV}.
Note that their reduction can be considered folklore.
However, Abrahamsen and Miltzow have carefully described all the details
and pointed out some important properties, which are crucial to us.
They describe several reductions, two of which are relevant to us.
In particular, they show that there is a reduction from \ETR to $\ETRAMI$, which
runs in linear time. We need the following definitions as a preparation,
see~\cite{DynamicToolboxETRINV}.

\begin{framed} 
	\noindent{\textbf \ETRAMI} \\
	\textbf{Input:} 
	A formula  
	$ \exists \; X_1,\ldots,X_n  \in \R: \Phi(X_1,\ldots,X_n)$ 
	over the reals, \\
	where  $\Phi$ is a conjunction of constraints of 
	the form: $X=1$ (introducing constant~1), $X+Y=Z$ 
	(addition), $X\cdot Y = Z$ (multiplication),
	 $Z\geq 0$ (inequality),
	with $X,Y,Z \in \{X_1,\ldots,X_n,Y_1,\ldots,Y_m\}$. \\
	\textbf{Question:} Is the input formula true?     
\end{framed}
We define \UETRAMI analogous to $\ETRAMI$,
except with universal and existential quantifiers, as
in $\UETR$.

Moreover,  given two sets $V\subseteq \R^n$ and $W\subseteq \R^m$, we say that $W$ is a
 \emph{linear extension} of $V$ if
 there is an orthogonal projection $\pi:W\longrightarrow V$
 and two vectors $a,b\in \Q^{n}$ 
 such that the mapping 
  \[x \mapsto a \odot \pi (x) + b\]
 is a continuous bijection.
 Here, $c \odot d $ 
 denotes the \emph{dot product} $(c_1d_1,\ldots,c_nd_n)$, for $c,d\in \R^n$.
 (Note that we require $a_i \neq 0$, for all $i$.)
 In this case we write $V\leqlin W$.
 
 Given a quantifier-free formula $\Phi(X_1,\ldots,X_n)$, we define the set
 \[V(\Phi) := \{ x \in \R^n : \Phi(x) \text{ is true}\}.  \]

In our context, it is important that a linear extension also
implies that the reduction can be applied to different
quantifiers. This is very useful, as it allows us to transfer reductions between \ER-hard problems
to the analogous \FER-hard problems.
\begin{lemma}
	\label{lem:Transition}
	Let 
	$\Psi = \left(\forall \, Y=(Y_1,Y_2,\ldots,Y_m)\right) 
 	\left(\exists \, X=(X_1,X_2,\ldots,X_n)\right) \colon 
 	\Phi(X,Y) ,$
 	and 
 		$\Psi' = \left(\forall \, Y=(Y_1,Y_2,\ldots,Y_m)\right) 
 	\left(\exists \, X=(X_1,X_2,\ldots,X_{n'})\right) \colon 
 	\Phi'(X,Y).$
 	If  $V(\Phi')$ is a \linearExtension of $V(\Phi)$ then 
 	$\Psi$ and $\Psi'$ have the same truth value.
\end{lemma}

\begin{proof}
	Let us assume that
	$\Psi$ is true. We have to show that $\Psi'$ is true as well.
	The reverse direction is similar. 
	Let us denote the bijection $f$ from $V(\Phi')$ to $V(\Phi)$,
	by 
	 \[f: z \mapsto a \odot \pi (z) + b,\]
	 for some $a,b\in \R^{m+n}$.
 	(Recall that $\pi$ denotes an orthogonal projection.)
	 Note that $f$ can be extended to a \emph{surjective} mapping  
	 $g: \R^{m+n'}\rightarrow \R^{m+n}$.

	 Now, let $y'\in \R^m$ be some arbitrary vector.
	 We have to show that there is a vector~$x'\in \R^{n'}$ such that
	 $\Phi'(x',y')$ is true. 
	 Let $(x^\ast,y) = g(0,y')$.
	 As we assume that $\Psi$ is true, there exists an $x$ such
	 that $\Phi(x,y)$ is true.
	 Because $f$ is a bijection, there exists  $(x',y'')\in V(\Phi')$ such that $f(x',y'') = (x,y)$.
	 Note that the projection $\pi$ is only acting on the dimension of
	 the existentially quantified variables.
	 The vectors $a$ and $b$ as in the definition of linear extension are just scaling 
	 and shifting.
	 Thus $y' = y''$.
	Consequently, it holds that $\Phi'(x',y')$ is true.
\end{proof}

\begin{theorem}[\cite{DynamicToolboxETRINV}]
 \label{thm:ETRAMI}
 For every instance $\Phi$ of \ETR, 
we can construct in $O(|\Phi|)$ time 
an instance $\Psi$ of \ETRAMI such that
  $V(\Psi)$ is a \linearExtension of $V(\Phi)$.
\end{theorem}
\begin{proof}
	This follows immediately from the proofs of Lemma~A and
	Lemma~C in~\cite{DynamicToolboxETRINV}.
	(Note that Lemma~C guarantees to preserve compactness. However, compactness is not used in any other way. Therefore, we do not need to require compactness.)
\end{proof}

\subsubsection{The proof of~\cref{thm:UETRNEW}}
We are now ready to prove~\cref{thm:UETRNEW}:
\begin{proof}[Proof of~\cref{thm:UETRNEW}.]
	Membership of \UETRNew to \FER follows from the definition of \UETR,
	as \UETRNew is a special case of \UETR.
	
	In order to show hardness, we reduce \UETR to \UETRNew.
	By~\cref{thm:ETRAMI} and~\cref{lem:Transition},
	it is sufficient to reduce from \UETRAMI.
	The two algorithmic problems \UETRAMI and \UETRNew
	are different in two aspects.
	Firstly, \UETRNew has no constraint of the form $X\geq 0$ .
	Secondly, \UETRNew has only positive variables.
	
	\paragraph*{Removing Inequalities.} 
	Let $X\geq 0$ be some inequality-constraint.
	It is obvious that $X$ is existentially quantified,
	as the formula otherwise, is trivially false.
%	(Not all real numbers are non-negative.).
	For every inequality, we do the following standard trick.
	We add another existentially quantified variable $V$ and 
	replace the old constraint by the new constraint $X=V^2$.
	It is obvious that this transformation can be done
	in linear time and is a proper reduction.
	
	\paragraph*{Changing ranges of quantifiers.}
	Next, we want to exchange quantifiers ranging over all reals with quantifiers ranging over positive reals.
	For each variable $Z$, we introduce two positive variables $Z^+$ and $Z^-$.
	If $Z$ is universally quantified, then so are both $Z^+$ and $Z^-$; analogously in the case if $Z$ is existentially quantified.
	Every appearance of~$Z$ is substituted by $(Z^+ - Z^-)$.
	It is easy to observe that the constructed formula is equivalent. 
	However, the form of the constraints, as
	defined above, is lost. 
	
	\paragraph*{Restoring the form of constraints.}
	Now, we want to restore the constraints. 
	We introduce new variables and, by following Abrahamsen and Miltzow~\cite{DynamicToolboxETRINV}, denote them 
	by multi-character symbols.
	For example, in order to introduce the square of an already established variable~$X$, we represent the new variable by $\const{X^2}$. The value of  $\const{X^2}$ will be forced by appropriate constraints.

	We introduce the existential variable $\const{1}$ with the constraint $\const{1} = 1$.
	Then, every constraint is transformed in the following way:
	A constraint to introduce a constant $(Z^+-Z^-)=1$ is transformed into $Z^+ = Z^- + \const{1}$.
	
	An addition constraint $(X^+-X^-) + (Y^+-Y^-)=(Z^+-Z^-)$ is equivalent to the following expression: $X^+ + Y^+ + Z^- = X^- + Y^- + Z^+$. We introduce new positive, existentially quantified variables and the constraints:
	\begin{align*}
	X^+ + Y^+ =& \const{X^+ + Y^+} \\
	X^- + Y^- =& \const{X^- + Y^-} \\
	\const{X^- + Y^-} + Z^+ =& \const{X^- + Y^- + Z^+} \\
	\const{X^+ + Y^+} + Z^- =& \const{X^- + Y^- + Z^+}.
	\end{align*}
	A multiplication constraint $(X^+-X^-) \cdot (Y^+-Y^-)=(Z^+-Z^-)$ is equivalent to the  expression $X^+Y^+ + X^-Y^- + Z^- = X^+Y^- + X^-Y^+ + Z^+$. We introduce new positive, existentially quantified variables and the constraints for each pair $\circ,\times\in\{+,-\}$
	\begin{align*}
	X^\circ \cdot Y^\times =& \const{X^\circ Y^\times} 
	\end{align*}
	as well as the constraints:
	\begin{align*}
	\const{X^+Y^+} + \const{X^-Y^-} =& \const{X^+Y^+ + X^-Y^-} \\
	\const{X^+Y^-} + \const{X^-Y^+} =& \const{X^+Y^- + X^-Y^+}\\
	\const{X^+Y^-+X^-Y^+}+Z^+	=& \const{X^+Y^- + X^-Y^+ + Z^+}  \\
	\const{X^+Y^- + X^-Y^+} + Z^- =& \const{X^+Y^- + X^-Y^+ + Z^+}.
	\end{align*}
	
	Note that now all constraints are of desired forms and all variables 
	can be assumed to be strictly positive.
	Specifically, newly introduced variables are the sum or
	product of other variables.
	Furthermore, the last three steps can all be done in linear time.
%	This completes the proof.
\end{proof}

\subsection{A Cook-Levin-type Theorem}
\label{sec:Cook-Levin}

Complexity classes are defined in various ways.
Some complexity classes are defined by an algorithmic
problem and a notion of equivalence or reduction; other complexity classes are described in terms of
models of computation. 
Most famously, the complexity class \NP is originally defined
in terms of a non-deterministic Turing machine.
Namely, the complexity class \NP consists of all problems that
can be solved in polynomial time on a non-deterministic Turing machine.
Alternatively, \NP can be defined as the set of all problems that reduce to boolean satisfiability in
polynomial time.
Cook and Levin independently showed that the two definitions are equivalent~\cite{Cook,levin1973universal}.
Boolean satisfiability plays a crucial role in showing \NP-hardness.
On the other hand, the formulation in terms of non-deterministic Turing machines
plays a crucial role to show \NP-membership.

The history of \ER is reversed. 
The complexity-class~\ER was defined in terms of the algorithmic
problem \ETR, which is the real analog to boolean satisfiability.
Recently, Erickson, van der Hoog, and Miltzow~\cite{RobustComputation}
gave a new definition of the complexity class~\ER in
terms of a model of computation, which is analogous to non-deterministic Turing machines.
They showed that the two definitions are equivalent.
In this way their result can be seen as a  Cook-Levin-type theorem for \ER.

In this section, we prove the Cook-Levin-type result for \FER yielding an alternative way to establish \FER-membership.
Roughly speaking, we show that \UETR-formulas and 
\emph{real challenge-verification} algorithms are equally powerful.
Simultaneously, this gives an alternative perspective
and definition on \FER. 

Our proof strongly relies on the result by Erickson, van der Hoog, and Miltzow~\cite{RobustComputation}.
To describe our results, we  follow their terminology. 
Their definition of the real-RAM corresponds to the intuitive definition most researchers use.
The real-RAM consists of real registers, word registers, 
a program counter and a central processing unit (CPU).
An algorithm is a list of supported instructions,
see~\cite{RobustComputation} for details. 
Inputs are denoted by $(x,y) \in \R^{n}\times \Z^{m}$.
A \emph{discrete decision problem} is a function $Q$ 
from arbitrary integer vectors to the booleans 
$\{\textsc{True},\allowbreak \textsc{False}\}$ 
(or equivalently, any language over the alphabet $\{0,1\}$).  
An integer vector $I$ is a \emph{yes-instance} of $Q$ 
if $Q(I) = \textsc{True}$ and 
a \emph{no-instance} of $Q$ if $Q(I) = \textsc{False}$.  Let $\circ$ denote the concatenation operator.
A \emph{real challenge-verification  algorithm} (CV algorithm) for $Q$ is a real-RAM
program~$A$ that satisfies the following conditions, for some constant~$c\ge 1$:
\begin{itemize}[noitemsep]
	\item
	$A$ halts after at most $N^c$ time steps, using word size $\lceil c\log_2 N \rceil$, given any input of total length~$N$.
	
	\item
	For every yes-instance $I\in\Z^n$, 
	and for every real vector $y$,
	there is a real vector~$x$ and an integer vector $z$, 
	each of length at most $n^c$, 
	such that $A$ \textsf{accept}s input $(y \circ x, I\circ z)$.
	
	\item
	For every no-instance $I$, 
	there is a real vector $y$
	such that for every real vector $x$ and an integer vector $z$, 
	each of length at most $n^c$, 
	$A$ \textsf{reject}s input $(y \circ x, I\circ z)$.
\end{itemize}
Note that we assume that the algorithm $A$ knows, which parts 
belong to $x$, $y$, $I$, and $z$ due to their length.

There exists a clear analogy to \UETR-formulas.
The instance $I$ represents the non-quantified part
of the formula, $y$ represents values of universally
quantified variables and $x$ represents existentially 
quantified variables. 
We call $y$ the \emph{challenge} and $x$ the \emph{witness}.
We are now ready to state the main theorem of the section.	
\begin{theorem}[Cook-Levin Analog]\label{thm:cookLevin}
	For any discrete decision problem $Q$ there is a real challenge-verification algorithm
	if and only if $Q\in \FER$.
\end{theorem}
\begin{proof}[Sketch of proof.]
	We have to show how an \UETR-formula can be transformed into
	an instance $I$ and a CV algorithm $A$ and vice versa. 
	Given an \UETR-formula~$\varphi$, it is very easy to simulate
	it using a CV algorithm, by merely evaluating the formula.
	The challenge and witness represent the universal and existential
	variables respectively.
	
	Given an instance $I$ and a CV algorithm $A$, we
	construct an equivalent \UETR-formula~$\varphi$.
	Interestingly, this can be done in the same way as 
	in Theorem~$10$ of Erickson, van der Hoog and Miltzow~\cite{RobustComputation}.
	To explain this, first note that the proof idea of the mentioned
	Theorem~10~\cite{RobustComputation} is to simulate every 
	execution step
	of the algorithm by an \ETR-formula. 
	However, the simulation does not mind
	how the variables are quantified. Thus, even if some of the 
	variables are universally quantified, the simulation is
	still identical.
\end{proof}

In the next section, we use \cref{thm:cookLevin} to prove the containment of area-universality in \FER (\Cref{pro:Containment}).
By the above result, this proof does not require to construct a formula, but solely
a real challenge-verification algorithm.
This allows to build on existing 
algorithms without reformulating them into the language of \FER-formulas.

\subsection{Containment of \AU in \FER}
\label{sec:auinfer}

In this section, we present a proof of \cref{pro:Containment}.

\AreaUniInFER*

\begin{proof}
	By \cref{thm:cookLevin}, it suffices to show that there exists a real challenge-veri\-fi\-ca\-tion algorithm for testing both \AUpos and \AU.
	Given a plane graph $G$, we encode it by its rotation systems (that is the cyclic order of incident edges at each vertex) and its outer face, i.e., by a vector of words.
	This defines our input $I$.
	Universally quantified areas define the challenge $y$.
	The witness $d$ consists of the vertex coordinates of a  realizing straight-line drawing~$D$ of~$G$.
	The algorithm checks that the drawing $D$ is a realizing drawing. \medskip
	
	For \AUpos this is straight-forward. In particular, we ensure that $D$ is 
		\begin{enumerate}
		[label=(\roman*),topsep=1pt,itemsep=-1pt,partopsep=1pt,parsep=1pt,
		leftmargin=0.8cm]
		\item planar,
		\item equivalent to $G$, and
		\item
		realizes the face areas defined in $y$.
	\end{enumerate}

In order to test planarity, we check for every pair of disjoint edges, whether the intersection of their segments is empty. 
To test the equivalence, we check if the rotation system of $G$ and $D$ coincide, i.e., if the cyclic order of incident edges are equal. Here we use the fact that the rotation system uniquely defines the 2-cell embedding of a planar graph~\cite[Theorem 3.2.4.]{mohar}.

In order to check the face areas, we exploit the fact that after checking planarity, we know that  every face corresponds to a simple polygon; in the presence of cut-vertices, the polygon is degenerate.
The area of a  (possibly degenerate) polygon can be computed with the {\em shoelace formula} which first described by 
Meister~\cite{meister1769generalia}. 
Denote the vertex coordinates of 
a simple polygon~\P in counterclockwise order by $(x_1,y_1),\ldots,(x_n,y_n)$. 
Then the area~$A$ of \P is given by
\begin{align}
2\cdot A(\P) =
\sum_{i=2}^{n-1} 
\det \begin{pmatrix}
x_1 & x_i & x_{i+1} \\
y_1 & y_i & y_{i+1} \\
1 & 1 & 1 \label{form:shoelace}
\end{pmatrix}.
\end{align}

	Proving containment of \AU is more intricate. We check that $D$
	\begin{enumerate}
		[label=(\roman*),topsep=1pt,itemsep=-1pt,partopsep=1pt,parsep=1pt,
		leftmargin=0.8cm]
		\item \label{itm:Crossing-Free}
		is crossing-free,
		\item \label{itm:AREAS}
		realizes the face areas defined in $y$, and
		\item \label{itm:totalArea} the complement of the outer face has the correct total area, namely the sum of the areas in $y$.
	\end{enumerate}
Furthermore, we show that the properties  \ref{itm:Crossing-Free},  \ref{itm:AREAS}, and \ref{itm:totalArea}  imply that $D$ is equivalent to~$G$, i.e., the set of inner faces coincide.

We show how the three steps can be computed with a real-RAM algorithm.
For \ref{itm:Crossing-Free}, the difficulty arises from the fact that we allow for degenerate drawings.  To check this property, we exploit a result of Atikaya, Fulek, and Tóth~\cite{weekEmbeddings};  restated in the terminology of this paper it reads as follows.

\begin{claim*}[\cite{weekEmbeddings}, Corollary 1.3]
	Given an abstract graph with $n$ vertices and a vector of vertex coordinates $d$,
	we can decide in $O(n^2\log n)$ time whether the  straight-line drawing  induced by $d$ is crossing-free.
\end{claim*}
	
For the remainder, we exploit the fact that after checking \ref{itm:Crossing-Free}, we know that  every face corresponds to a (possibly degenerate) simple polygon.
Similarly to the above, we check properties~\ref{itm:AREAS} and  \ref{itm:totalArea} by computing the area of a  (possibly degenerate) polygon with \cref{form:shoelace}.

It remains to show  how \ref{itm:Crossing-Free},  \ref{itm:AREAS}, and \ref{itm:totalArea}   imply that $D$ and $G$ have the same set of faces, i.e., that $D$ is a crossing-free drawing of~$G$.
Observe that we verify~\ref{itm:AREAS}, using \cref{form:shoelace}.
In other words, we ensure that the \emph{signed areas} are correct. 
In particular, this implies that the \emph{orientation} of each face $D$ is correct, i.e., when walking in counter-clockwise direction on the boundary of a face in $D$, the bounded part lies on the left side. 
If the bounded part were to lie on the right side, the determinant in \cref{form:shoelace} would have a negative sign.

For each inner face $f$ of $G$, we consider the corresponding face polygon in $D$ defined by the bounded region of the face cycle.
Due to the correct orientation of all face polygons, every inner edge of~$G$ is covered from both sides in~$D$. 
Moreover, every outer edge of~$G$ is covered from the inside in~$D$. In particular, the region obtained by the union of all face polygons has no holes and is bounded by the outer edges of $G$ forming a closed curve. 
We claim that the face polygons are interior-disjoint; 
otherwise some parts of the plane are covered twice -- a contradiction to the fact that the area of the face polygons sums to the area of the complement of the outer face.
Because the face polygons are interior-disjoint,  they are faces of~$D$ and hence, the drawing $D$ is a crossing-free drawing for $G$.
This completes the proof.
\end{proof}

\section{Hardness of \PA}\label{sec:PA}
In this section, we prove the following theorem.
\triples*

\begin{proof}
The membership follows from the fact that we can easily express the area of a triangle by a polynomial equation, see \cref{form:shoelace} in \cref{sec:auinfer}.
So an instance of \PA can be expressed as an \UETR-formula, which is a conjunction of equations of the above form for each triple. 

For the hardness, we reduce from $\UETRNew$. For every instance $\Psi$ of \UETRNew, we give a set of points $V$ and unordered triples~$T$, along with a partial area assignment $\cal A'$.
Let $\Psi$ be a formula of the form: 
\[\Psi = (\forall Y_1, \ldots,Y_m \in \R^+ )(\exists X_1,\ldots,X_n \in \R^+) \colon \Phi(Y_1, \ldots,Y_m,X_1,\ldots,X_n).\]
Recall that $\Phi$ is a conjunction of constraints of the form $ X = 1$, $X+Y = Z$, and $X \cdot Y = Z$.
First, we show how to express $\Phi$.   
Our gadgets are similar to the \emph{von Staudt constructions}, used for showing $\ER$-hardness of \textsc{Order Type} (see Mn\"ev~\cite{Mnev}  or Richter-Gebert~\cite{Richter-Gebert95}).
All variables are represented by points on one line which we denote by $\ell$ for the rest of the proof.
First, we enforce points to be on~$\ell$. Afterwards, we construct gadgets for mimicking addition and multiplication. 
Finally, we describe how to represent universally quantified variables.

We introduce three points $p_0$, $p_1$, and $r$ and define $\A'(p_0,p_1,r):=1$. 
The positive area ensures that the points are not collinear and pairwise different.

Denoting a line through two points $a$ and $b
$ by $\ell_{a,b}$, 
we set $\ell := \ell_{p_0,p_1}$.
The points on $\ell$ will correspond to real numbers, where $p_0$ and $p_1$ are interpreted as 0 and 1, respectively.
Each variable~$X$ is represented by a point $x$ on $\ell$.
Additionally, since all variables $X$ are non-zero, we introduce a triangle forcing~$x$ to be different from $p_0$. 
In general, we can ensure that two points 
$x_1,x_2$ are distinct, by introducing a point 
$q$ and adding a triangle $(x_1,x_2,q)$ with $\A'(x_1,x_2,q):=1$.
The absolute value of $X$ is defined by $\|p_0x\|$; if $x$ and $p_1$ lie on the same side of $p_0$, 
then the value of $X$ is positive, otherwise it is negative. Here, we allow negative values, but later we force the original variables to be positive.

A simple way to force a point $x$ on $\ell$, is to set
$\A'(x,p_0,p_1):=0$. However, as we want to avoid prescribing area 0, we introduce a more sophisticated way.
As an important building block, 
we show how to enforce on four pairwise 
different points $t,t',b,b'$ that the segment ${tt'}$ is 
parallel to the segment ${bb'}$ and has half its length, see \cref{fig:TrapezoidA}; besides these two properties no additional constraint on any of the four points is imposed.
Thus, we define 
$\A'(t,t',b)=\A'(t,t',b')=1$ and 
$\A'(b,b',t)=\A'(b,b',t')=2$, see \cref{fig:TrapezoidB}.

\begin{figure}[htb]
	\centering
	\begin{subfigure}[b]{.4\textwidth}
		\centering
		\includegraphics[page=11]{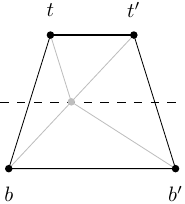}
		\caption{} 
		\label{fig:TrapezoidA}
	\end{subfigure}\hfil
	\begin{subfigure}[b]{.4\textwidth}
		\centering
		\includegraphics[page=12]{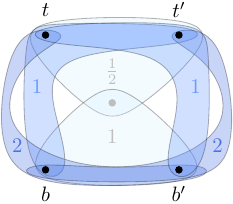}
		\caption{} 
		\label{fig:TrapezoidB}
	\end{subfigure}
	\caption{The gadget to force parallel segments. }
	\label{fig:Trapezoid1}
\end{figure}

We show that $b$ and $b'$ lie on the same side of the line $\ell_{t,t'}$:
Suppose for  contradiction that $\ell_{t,t'}$ separates $b$ and $b'$. 
If, additionally, $t$ and $t'$ are on the same side of $\ell_{b,b'}$ then the triangle $(b,b',t)$ is contained in or contains 
the triangle $(b,b',t')$, see \cref{fig:TrapezoidC}. 
However, both triangles have the same area and $t \neq t'$, which is a contradiction.
Consequently,  $\ell_{b,b'}$ separates $t$ and $t'$ and the quadrangle $tbt'b'$ can be partitioned by either diagonal $bb'$ or $tt'$ as illustrated in \cref{fig:TrapezoidD}. 
Thus,
$2=\mathcal A (t,t',b)+
\mathcal A (t,t',b')=\mathcal A (b,b',t')
+\mathcal A (b,b',t)=4$, which is again a contradiction. Therefore  $b$ and $b'$ lie on the same side of $\ell_{t,t'}$.

\begin{figure}[htb]
	\centering
	\begin{subfigure}[b]{.45\textwidth}
		\centering
		\includegraphics[page=13]{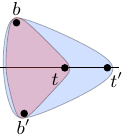}
		\caption{Case 1: $t,t'$ are on the same side of $\ell_{b,b'}$.} 
		\label{fig:TrapezoidC}
	\end{subfigure}\hfil
	\begin{subfigure}[b]{.45\textwidth}
		\centering
		\includegraphics[page=14]{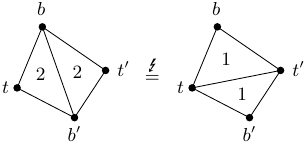}
		\caption{Case 2: $t,t'$ are separated by $\ell_{b,b'}$.} 
		\label{fig:TrapezoidD}
	\end{subfigure}
	\caption{The situation that $\ell_{t,t'}$ separates $b$ and $b'$.}
	\label{fig:Trapezoid2}
\end{figure}

By the prescribed area, $b$ and $b'$ have the same distance to $\ell_{t,t'}$, i.e., the segments $tt'$ and $bb'$ 
are parallel, and $bb'$ has twice the length of $tt'$. Moreover, no further constraints are imposed on $t,t',b,b'$.

We use this building block in two ways. 
First, the gadget is used to force points on a specific line defined by two distinct points. Note that a point $x$ with $\A'(t,t',x):=\nicefrac{1}{2}$, $\A'(b,b',x):=1$ is forced to lie on the line cutting the trapezoid at half height as illustrated in \cref{fig:Trapezoid1}. 
In order to force points on the line $\ell$, we use a copy of this building block as follows: Introduce four new distinct points $c,c',d,d'$. Construct a trapezoid with parallel segments $cc'$ and $dd'$, where the length of the second one doubles the length of the first one.
Let $P$ denote a set of points to be enforced on $\ell$, then we define for every $x\in P\cup\{p_0,p_1\}$ the areas
$\A'(c,c',x):=\nicefrac{1}{2}$ and $\A'(d,d',x):=1$. This introduces no other constraints on the position of~$x\in P$. Likewise, we can introduce points on any given line defined by two distinct points. We refer to this as the \emph{line construction}.

The second way to use the trapezoid construction is to enforce two lines to be parallel. For this, we first introduce new and distinct points on the two lines, as described in the previous paragraph.
Afterwards, we build the trapezoid gadget on the four points. We refer to this as the \emph{parallel-line construction}.

Now, we describe the \emph{addition gadget} for a constraint $X+Y=Z$. Let $x,y,z$ be the points encoding the values of $X,Y,Z$, respectively. Recall that  $x,y,z\in\ell$ and  $x,y,z\neq p_0$. 
We introduce a point $q_1$ and prescribe the areas 
$\A'(p_0,x,q_1)=\A'(y,z,q_1)=1$, see \cref{fig:parallelLines2}.
Because the two triangles have the same height and same area, it holds that $\|yz\|=\|p_0x\|$. Thus, the value of $Z$ 
is either $X+Y$ or $Y-X$. Analogously, we introduce a point $q_2$ and define 
$\A'(p_0,y,q_2)=\A'(x,z,q_2)=1$, implying that $Z$ is either $Y+X$ or $X-Y$. 
Therefore either $Z=X+Y$ (the intended solution) or $Z=X-Y=Y-X$. 
The latter case implies $X=Y$ and thus $Z=0$. This 
contradicts the fact that $z\neq p_0$.

\begin{figure}[htb]
	\centering
	
	\begin{subfigure}[b]{.4\textwidth}
		\centering
		\includegraphics[page=3]{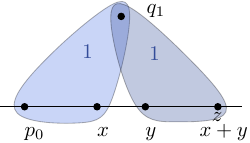}
		\caption{The addition gadget.}
		\label{fig:parallelLines2}
	\end{subfigure}\hfil
	\begin{subfigure}[b]{.4\textwidth}
		\centering
		\includegraphics[page=10]{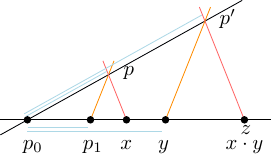}
		\caption{The multiplication gadget.}
		\label{fig:mult}
	\end{subfigure}
	\caption{The constraint gadgets.}
	\label{fig:gadget}
\end{figure}

To construct a \emph{multiplication gadget} for a
constraint $X \cdot Y = Z$, let $x,y,z(\neq p_0)$
 be the points encoding the values of $X,Y,Z$, respectively.
We introduce two distinct points $p,p'$ whose supporting line contains $p_0$, i.e., $p_0,p,p'$ lie on a common line by the line-construction.
Furthermore, we make sure that $p$ and $p'$ are not on $\ell$.
To do so, we introduce two new points $q,q'$, force them to lie on $\ell$ by the line construction, and prescribe the areas $\A'(p_0,p,q):=1$ and $\A'(p_0,p',q'):=1$. Note that this does not introduce any additional constraints on the positions of $p$ and $p'$.
By the parallel-line construction, we enforce the line $\ell_{p_1,p}$ to be parallel to $\ell_{y,p'}$ and the line $\ell_{x,p}$ to be parallel to $\ell_{z,p'}$,
see~\cref{fig:mult}. 
The intercept theorem, which is also known as Thales' theorem or the basic proportionality theorem~\cite{wikiIntercept}, ensures that the following ratios 
coincide:
$$|p_0p|/|p_0p'|=|p_0p_1|/|p_0y|=|p_0x|/|p_0z|.$$ 
This also holds for negative values.
By  definition of $x,y,z$, we obtain $1/Y = X/Z$, and hence $X \cdot Y=Z$.
Recall that $p_1 = 1$.

For every universally quantified variable $Y_i$, $i\in\{1,\dots, m\}$, let $y_i$ be the point encoding 
the value of $Y_i$ with $y_i\in \ell$ and
$y_i \neq p_0$. We introduce a triple 
$f_i = (p_0,r,y_i)$, whose area is 
universally quantified. Recall that $r$ 
is a point with $\A'(p_0,p_1,r)=1$. 
To enforce each original variable $X$ to be positive, we add an existentially quantified variable~$S_X$ and 
the constraint $X = S_X \cdot S_X$ where $S_X$ may or may not be positive. 
This finishes  the polynomial time reduction. 

It remains to argue that $\Psi$ is true if and only if for our constructed instance of \PA, where $\A'$ is the partial assignment of areas, every assignment $\A$ consistent with $\A'$ is realizable.
Suppose $\Psi$ is true, let $\A'$ be as above, and consider an assignment $\A$ that is consistent with $\A'$.
Let $V(Y_i)$ be the values assigned to the triples $f_i$, and let $V(X_i)$ be the values of the variables $X_i$ in some satisfying assignment for $\Phi$.
Let $y_1,\dots,y_m,x_1,\ldots,x_n$ be points positioned on a line at distances from a point $p_0$ corresponding to these values.  Since each addition and multiplication relation specified by $\Phi$ holds, the corresponding  gadgets can be realized, so \A is realizable.
Suppose now that every assignment $\A$ that is consistent with $\A'$ is realizable, and consider values $V(Y_1),\ldots,V(Y_m) \in \R^+$ of the universally quantified variables of $\Psi$.
Then, there is a realization of $\A$ where each triple $f_i$ has area $V(Y_i)$, in this realization $p_0 \neq p_1$, and $V(X_i) = \|x_i -p_0\|/\|p_1-p_0\|$ is a satisfying assignment for $\Phi$.  Thus $\Psi$ is true.

Finally, let us point out that for each variable and for each gadget we introduced a constant number of vertices and a constant number of triangles, so the number of triangles is linear in the number of vertices.
\end{proof}

Let us point out that if we start the reduction in \cref{thm:hyper} from a variant of \UETRNew, where universally quantified variables are non-negative (instead of positive), we can obtain the hardness of a variant of \PA, where we allow degenerate triangles (i.e., of area 0).
Establishing \FER-hardness of such a variant of \UETRNew can be done in a way analogous to \cref{thm:UETRNEW}.

We conclude this section with a small discussion of how to strengthen this result. 
An interesting step towards establishing our conjecture would be to prove an analog of \cref{thm:hyper}, where the realizing point set and the triples induce a plane triangulation.
However, the construction in our proof of \cref{thm:hyper} is highly non-planar, and obtaining such a result seems to require a very different approach.

\section{Hardness of \PPA}\label{sec:PPA}
In this section we prove \cref{thm:PPA} by presenting 
a reduction from \PETRNew. 

\PPAthm*
\begin{proof}
Let $\Psi = \exists X_1\ldots X_n:\Phi(X_1,\ldots,X_n)$ be an instance of \PETRNew.
Recall that we can assume that if $\Psi$ is a 
{\textsc yes}-instance, then it has a solution, 
in which the values of variables are in the 
interval $(0,\lambda)$ with $\lambda = 5$. 
We construct a plane graph $G_\Psi=(V,E)$, a (positive)
area assignment $\mathcal A$ of inner 
faces of $G_\Psi$, and fixed positions of a 
subset of vertices, such that $G_\Psi$ has 
a realizing drawing respecting pre-drawn vertices 
if and only if $\Phi$ is satisfiable by real values 
from the interval $(0,\lambda)$. 
  
Consider the incidence graph of $\Phi$.
It is straightforward to verify that a planar drawing of this graph can be modified to an embedding in which 
\begin{itemize}
	[topsep=1pt,itemsep=-1pt,partopsep=1pt,parsep=1pt,
	leftmargin=0.8cm]
\item each edge is represented by a polyline consisting of vertical and horizontal segments,
\item the polylines of edges sharing a vertex might partially overlap -- the common part, called a bundle, is a polyline that contains the common vertex,
\item each vertex representing a variable is incident to one bundle of edges.
\end{itemize}
 \autoref{fig:incidenceGraph} depicts an example of such an embedding.
We will construct the instance of \PPA, based on the formula $\Psi$ and its incidence graph.

\begin{figure}[hbt]
	\centering
	\includegraphics{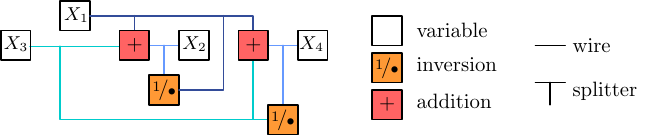}
	\caption{The embedding of the incidence graph of the formula~$\Psi = (X_1 + X_2 = X_3)\land (X_1 \cdot X_2 =1)\land(X_1 +  X_4 = X_3)\land (X_4 \cdot X_3=1)$.}
	\label{fig:incidenceGraph}
\end{figure}

We design several types of gadgets: 
{\em variable gadgets} representing variables, 
as well as {\em inversion} and {\em addition gadgets}, 
realizing the corresponding constraints. Moreover, 
we construct {\em wires} and {\em splitters} in 
order to copy and transport information.
Some vertices in our gadgets will have prescribed positions. 
We call such vertices \emph{fixed} and we call all 
other vertices  \emph{flexible}.

The \emph{variable gadget} for a variable $X$ consists of four fixed vertices $a$, $b$, $c$, $d$ and one flexible $v_x$, see the left image of \autoref{fig:var}. The fixed vertices $a,b,c,d$ form a square of area $1$.
Observe that the prescribed area of the face $a,d,v_x,c,b$ and planarity constraint force $v_x$ to be on the segment $cd$. 
Now the length of the segment $cv_x$ specifies the value of $X$, divided by the scaling factor $\lambda$. We symbolized this segment by a bold gray line in \autoref{fig:var}.

\begin{figure}[hb]
	\centering
	\begin{subfigure}[b]{.4\textwidth}
		\centering
		\includegraphics[page=2]{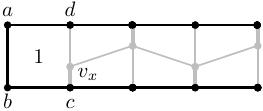}
		\caption{} 
		\label{fig:var}
	\end{subfigure}\hfill
	\begin{subfigure}[b]{.3\textwidth}
		\centering
		\includegraphics[page=4]{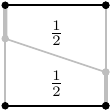}
		\caption{} 
		\label{fig:wire}
	\end{subfigure}\hfill
	\begin{subfigure}[b]{.2\textwidth}
		\centering
		\includegraphics[page=5]{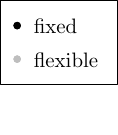}
	\end{subfigure}
	\caption{(a) The variable gadget with an attached wire and (b) a wire gadget.}
\end{figure}

Now, we introduce a constant 1 with a constraint $X=1$ by using a variable gadget. Instead of making vertex $v_x$ free, we make it fixed and place it in such a way that the distance from $c$ to $v_x$ is $\nicefrac{1}{\lambda}$.
We will sometimes assume that a flexible vertex is forced to lie on a specified segment because this property can be induced by a variable gadget.

The \emph{wire gadget} 
consists of several box-like fragments: four fixed vertices positioned as the corners of an axis-parallel unit square  and two opposite fixed edges, see \autoref{fig:wire}. Each of the other two sides of the square is subdivided by a flexible vertex, these two vertices are joined by an edge. Each of the two quadrangular faces has a prescribed area of $\nicefrac{1}{2}$.
  
Note that if one of the flexible vertices is collinear with its fixed neighbors, so is the other one. The crossing-freeness constraint ensures that each flexible vertex lies between the corresponding fixed vertices, as otherwise the edge that joins them would cross one of the edges on the boundary of the gadget.
Moreover, in a sequence of such squares, the segment representing the value of the variable alternates between top and bottom. Hence, if necessary, we may use an odd number of fragments in order to invert the side where the value is represented. 
Wires are used to connect variable gadgets to inversion and addition gadgets, see \autoref{fig:incidenceGraph}.
Let us point out that a variable gadget may be connected to other gadgets only using a wire.
This ensures that the property that we stated when introducing a variable gadget: the flexible vertex does not lie outside the square forming the gadget, see \autoref{fig:var}.

If one variable participates in several constraints, we need to be able to split the wire, in order to provide a connection to all necessary gadgets. The \emph{splitter gadget} contains a central fixed square of area 1, where each side is adjacent to a triangle of area $\nicefrac{1}{2}$, see \autoref{fig:split}.
   \begin{figure}[htb] 
   	\centering
   	\begin{subfigure}{.4\textwidth}
   		\centering
   		\includegraphics[page=6]{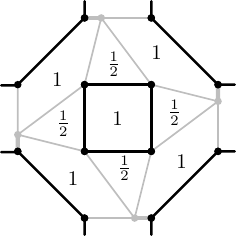}
   		\caption{}
   		\label{fig:split}
   	\end{subfigure}   
   	\hfil
   	\begin{subfigure}{.4\textwidth}
   		\centering
   		\includegraphics[page=7]{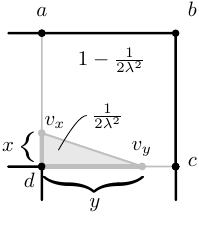}
   		\caption{}
   		\label{fig:inv}
   	\end{subfigure}   
   	\caption{(a) The splitter gadget and (b) the inversion gadget.}
   \end{figure}
   Each triangle fixes a flexible vertex on a line. These flexible vertices and their neighbors on the boundary of the splitter gadget are identified with the appropriate vertices in a a wire; this is how we connect the splitter with other gadgets.
   Observe that the value of one variable fixes the values of all variables. Indeed, since the non-fixed faces of the gadget are arranged in a circular way, it is straightforward to verify that all triangular faces must have exactly the same shape (up to translation and rotation).
   Note that each face  area $1$ has exactly two flexible vertices which are forced to lie on specific segments. Thus if one of them is determined, clearly the other is also uniquely determined in a realizing drawing.  Note that we may use the splitter gadget not only for splitting wires, but also for realizing turns. 

Now let us discuss gadgets for inversion and addition constraints.
First consider an inversion constraint $X \cdot Y = 1$.
The \emph{inversion gadget} consists of four fixed vertices $a,b,c,d$, two flexible vertices $v_x$ and $v_y$, and the edge $v_xv_y$, see~\cref{fig:inv}. The fixed vertices belong to a unit square. By linking the inversion gadget with wires, we can ensure that $v_x$ belongs to the segment $ad$ and $v_y$ belongs to the segment $cd$.
Let $x$ be the distance from $d$ to $v_x$, and $y$ be the distance of $d$ to $v_y$.

Suppose that these distances represent, respectively, the values of variables $X$ and $Y$ as described in the paragraph about variable gadgets; this is obtained by identifying them with appropriate vertices of such a gadget or a wire. Namely, the value of $X$ is $\lambda\cdot x$, and the value of $Y$ is $\lambda\cdot y$. To realize the area $\nicefrac{1}{(2\lambda^2)}$ of the triangle $v_xv_yd$, the lengths $x$ and $y$ are forced to satisfy $\nicefrac{xy}{2} = \nicefrac{1}{(2\lambda^2)}$. This implies that the value of  $X \cdot Y$ is $\lambda x \cdot \lambda y = 1$. Consequently, the inversion constraint is satisfied.

Finally, we define an \emph{addition gadget} for an addition constraint $X+Y=Z$.
We introduce fixed vertices $a,b,c,d,e,f,g,h,i,j,k,l$ placed on an integer grid, as depicted in \autoref{fig:add}. 
Moreover, flexible vertices $v_x,v_y,v_z$ encode the values of variables $X,Y,Z$, respectively;
 they are identified by appropriate flexible vertices in  wire gadgets. 
This ensures that $v_x$, $v_y$, $v_z$ are forced to lie on the segments $cg$, $fj$, and $k l$, respectively.
Finally, we introduce the flexible vertex $q$, joined to $d,e,i,$ and $h$. 
The areas of faces are set as follows:
\begin{align*}
f_1 &= v_x,c,a,b,f,v_y,e,q,d & \A(f_1) &= 3.5\\
f_2 &= v_x,d,q,h,v_z,k,g & \A(f_2) &=2\\
f_3 &= e,v_y,j,i,q & \A(f_3)&=1\\
f_4 &= v_z,l,i,q,h & \A(f_4)&=1.
\end{align*}

\begin{figure}[ht]  
 \begin{center}
   \includegraphics[page=5]{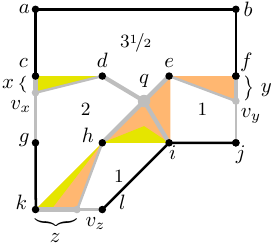}
   \caption{The addition gadget.}
   \label{fig:add}
   \end{center}  
\end{figure}

Let $x$ be the distance from $c$ to $v_x$, let $y$ be the distance from $f$ to $v_y$, and let $z$ be the distance from $k$ to $v_z$. 
Since $v_x, v_y, v_z$ encode the values 
of $X$, $Y$, and $Z$, respectively,  we know that $X = \lambda x$, $Y = \lambda y$, and $Z = \lambda z$.
We aim to show that the areas of faces of the gadget force that $x+y=z$.
Observe that the area of the triangle
 $c,v_x,d$ is $a_x := \nicefrac{x}{2}$, the area 
 of the triangle $e,f,v_y$ is $a_y := \nicefrac{y}{2}$ 
 and the area of $k,h,v_z$ is $a_z :=\nicefrac{z}{2}$.
Thus is it sufficient to show that $a_x + a_y = a_z$. 

The area of the triangle $d,e,q$ is forced 
to be $\nicefrac{1}{2} - a_x - a_y$, because of the 
constraint on the area of $f_1$.
Note if $\Psi$ is a positive instance, we may assume that 
$Z = X + Y < \lambda$, thus the area of the triangle $d,e,q$ 
is positive. 
On the other hand, if $X+Y > \lambda$, it is easy to see that 
there is no way to put flexible points to realize the prescribed areas.
Since the triangles $d,e,q$ and $h,i,q$ have 
parallel bases and share the third vertex, 
we observe that the area of the triangle 
$h,i,q$ is forced to be $a_x + a_y$.
It is easy to see that the area of $f_4$ equals 
$1 = A(h,k,l,i) + A(h,q,i) - A(h,k,v_z) = 1 + a_x + a_y - a_z.$
Thus we enforce $a_x + a_y = a_z.$

It remains to show that all faces have the 
correct area given $a_x+a_y = a_z$.
Observe that the area of the triangle $e,q,i$ is 
forced to be equal $a_y$, because of the constraint on the area of $f_3$.
Consequently, the triangle $d,q,h$ has area $\nicefrac{1}{2} - a_y$.
It holds that
\[\A(f_2) = A(c,d,h,k) - A(c,v_x,d) + A(d,h,q) + A(h,k,v_z) = \nicefrac{3}{2} -  a_x +(\nicefrac{1}{2}- a_y) + a_z = 2.\]

Finally, observe that by connecting the gadgets 
we might have introduced some inner faces, which 
do not have specified areas. However, the 
area of each gadget is fixed, so we exactly 
know what the areas of these newly 
introduced faces are. We finish the construction 
by fixing the area for each remaining face, 
according to the way it is drawn in the plane.
 
\enlargethispage{12pt}
Hence we created a planar graph $G_\Psi$ where the positions of some vertices are fixed, and 
a positive area assignment $\A$, such that $G_\Psi$ has 
an $\A$-realizing drawing that respects the fixed vertices if and only if $\Psi$ is 
true. The correctness of the construction 
follows easily from the correctness of individual gadgets.
Let us point out  that we may assume that all variables are in $(0,\lambda)$ and thus
 the flexible vertex of each variable gadget is not superimposed on another vertex;
this property propagates to wires and other gadgets.
Consequently, the realizing drawing is planar.
The size of $G$ is clearly polynomial in $|\Psi|$, 
so the proof is complete.
\end{proof}

%\newpage
\section{Discussion and Open Problems}\label{sec:Others}
In this section, our aim is two-fold. Firstly, we present further variants of \AUpos that might pave the path on the journey to prove \FER-hardness. Secondly, we want to motivate further research on the complexity classes \FER and \EFR by presenting  candidates of complete problems. In particular, we highlight connections to the concepts of imprecision, robustness, and extendability.  

\subsection{Open problems related to \AUpos}

Our first open problem aims at strengthening \cref{thm:hyper}.

\begin{q}\label{problem:innerTriangPartPres}
	Given a plane inner triangulation $T$ and an area assignment $\A'$ for some inner faces of $T$. Is it \FER-complete to decide whether there exist realizing drawings for all area assignments extending $\A'$?
\end{q}

Similarly, it is interesting to investigate the complexity of deciding whether or not an area assignment is realizable.
\begin{q} \label{problem:PAA}
	Is \PAA \ER-complete?
\end{q}
 Potentially, answering \cref{problem:innerTriangPartPres} also implies an answer for \cref{problem:PAA}.

%\linda{It is also interesting to study analogous questions in 3 dimension. Do we want to say something about this?}
%A closer 3-space analog to a plane triangulation would be a triangulation of a tetrahedron.  That is, an abstract simplicial complex realizable in $\R^3$ so that the union is a single tetrahedron who's faces also belong to the simplicial complex.
%
%
%\begin{q} \label{problem:PV-VU-ball}
%	Do \cref{thm:PV} and \cref{thm:VU} still hold if we additionally require the abstract simplicial complex to be a triangulation of a tetrahedron? 
%\end{q}

\subsection{Candidates for \FER- and \EFR-complete problems }
Complexity classes become interesting if they contain interesting algorithmic problems.
To motivate the research on \FER and \EFR, we present some candidates of problems that might be complete for these classes. 
We draw connections to existing concepts in computer science, like robustness and imprecision.
Although \FER and \EFR  are different complexity classes, it is worth mentioning that the complement of every language in \EFR belongs to \FER. Thus from an algorithmic point of view they are equally difficult.

The aim of this section is to point at potentially interesting future problems rather than giving technical insights.
We start with the problem we find most natural, besides \AU. Afterwards we come to the notions of universal extension problem, imprecision, and robustness. We conclude with a problem that has a similar flavor as area-universality but turns out to be polynomial time solvable.

A very natural metric for point sets is the 
so-called Hausdorff distance. For two sets 
$A,B\subseteq \R^d$, the Hausdorff distance 
$d_H(A,B)$ is defined as 
\[d_H(A,B) = \max \{ \sup_{a\in A} \inf_{b\in B} 
\|ab\|,  \sup_{b\in B} \inf_{a\in A}\|ab\|\},\]
where $\|ab\|$ denotes the Euclidean distance 
between the two points $a$ and $b$.
We define the corresponding 
algorithmic problem as follows:
\begin{framed}
\noindent {\bf{\textsc{Hausdorff Distance}}}\\
\noindent {\bf Input:} Quantifier-free formulas $\Phi$ and $\Psi$ in the first-order theory of the reals, $t\in \Q$.\\
\noindent {\bf Question:} Is $d_H(S_{\Phi},S_{\Psi})\leq t$?
\end{framed}
Recall that $S_{\Phi}$ is defined as $\{x\in \R^n : \Phi(x) \}$.
Marcus Schaefer~\cite{SchaeferHausdorf} pointed out the following interesting question: 
\begin{q} 
 Is computing the Hausdorff distance of semi-algebraic sets \FER-complete?
\end{q}

 \begin{figure}[tbhp]
 	\centering
  	\includegraphics{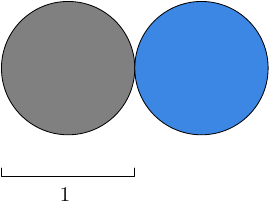}
 	\caption{The two disks have Hausdorff-distance $1$,
 	but $d$-distance $0$.}
 	\label{fig:DifferentDistances}
 \end{figure}
 
 Schaefer and \v{S}tefankovi\v{c}~\cite{DBLP:journals/mst/SchaeferS17} studied the 
 following very common notion of distance, 
 which we call here \emph{d-distance} for clarity. 
 The $d$-distance between two  sets $A$ and $B$ 
 in $\R^d$ is defined as 
 $d(A,B)= \inf\{\|ab\|: a\in A, b\in B\}.$ 
 See~\cref{fig:DifferentDistances} to see
 a simple example on how those two distances
 are different.
 They show the following lemma.
\begin{lemma}[\cite{DBLP:journals/mst/SchaeferS17}]
Deciding if two semi-algebraic sets have $d$-distance zero is \ER-complete.
\end{lemma}
Note that containment in \ER is not clear for the Hausdorff-distance. However, it is contained in \FER.
An encoding of Hausdorff-distance as a formula of 
the first order theory of the reals looks as follows:
\[(\forall a\in A, \varepsilon > 0 \ \exists b \in B :
 \|ab\|< \varepsilon + t )
 \  \land \
 (\forall b'\in B, \varepsilon > 0 \
 \exists a' \in A : \|a'b'\|< \varepsilon + t). \]
The formula holds true if and only if $d_{H}(A,B) \leq t$.
Note that we need the $\varepsilon$ in the formula,
because semi-algebraic sets can be open.
This can be easily reformulated as a formula 
in prenex form with two blocks of quantifiers.
By using the following logical equivalences: 
\begin{align*}
 \big(\forall x \exists\ y \colon \Phi(x,y)\big) \land \big(\forall x' \ \exists y' \colon \Phi'(x' ,y') \big)
&\equiv 
\forall x,x'\ \exists y,y' \colon \Phi(x,y) \land \Phi'(x' ,y') \text{ and }\\
\forall x\in X \ \exists y\in X : \Phi(x,y)
&\equiv 
\forall x \ \exists y : x\in X \Rightarrow \big( y\in Y \land \Phi(x,y) \big),
\end{align*}
where $\Phi$ and $\Phi'$ are two quantifier free formulas, and 
$\equiv$ indicates that the 
two formulas are logically equivalent.

By definition, a semi-algebraic set is a 
subset $S$ of $\mathbb R^n$ defined by a 
finite sequence of polynomial equations 
and strict inequalities or any finite union of such sets.

\paragraph*{Universal Extension Problems.}
Before we explain the general concept of an Universal Extension problem, 
we take a look at a specific example: the \AGP. In this problem we 
are given a simple polygon $P$ and we say that a point $p$ \emph{sees} 
another point $q$ if the line segment $pq$ is fully contained in $P$. 
The \AGP asks for a smallest set of points, which are called guards, 
such that every point inside the polygon is seen by at least one guard. 
It was recently shown that \AGP is \ER-complete~\cite{AbrahamsenAM18STOC}.
In an extension version of \AGP, we are given a polygon and a partial set of guards $G_1$. The task is to find a set of guards $G_2$ such that every point inside the polygon is seen by at least one guard from the set $G_1 \cup G_2$.
So the spirit of an extension problem is to give a partial solution as an additional input and ask if it is extendable to a full solution. 
Now we define the universal extension variant of the \AGP as follows. The input consists of a simple polygon $P$ and a set of regions $R_1,\ldots,R_t$ inside $P$. We ask whether for every guard placement of $G_1 = \{g_1,\ldots,g_t\}$ such that $g_i\in R_i$ it holds that there exists a second set of guards $G_2$ of some given size such that $G_1\cup G_2$ guard the entire polygon. We denote this as the \textsc{Universal Guard Extension} problem. For an example consider \autoref{fig:extGuarding}.

\begin{figure}[htbp]
 \centering
 \includegraphics[page=1]{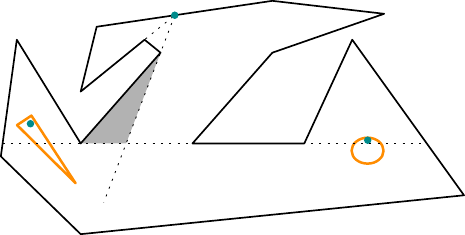}
 \caption{If two guards in the orange regions are badly placed it is impossible to guard the remaining polygon with just one guard.}
 \label{fig:extGuarding}
\end{figure}

\begin{framed}
\noindent {\bf \textsc{Universal Guard Extension}}\\ 
   \noindent {\bf Input:} A polygon $P$, regions $R_1,\ldots,R_t$ inside $P$, and a number $k\in \N$.\\   
    \noindent {\bf Question:} Is it true that for every placement of guards $G_1 = \{g_1,\ldots,g_t\}$ with $g_i\in R_i$, there exists a set $G_2$ of $k$ guards, such that $G_1\cup G_2$ guard $P$ completely?
 \end{framed}
 
 \begin{q}
  Is \noindent \textsc{Universal Guard Extension} \FER-complete?
 \end{q}

The spirit is that we do not want to extend one given partial solution (as in the extension variant), but all possible partial solutions. In this example we require the guards to lie in specific regions. This is necessary, as otherwise we had to consider such extreme cases as a partial solution with all the guards in a single point, which would not be meaningful. As we will later see it seems to be a common theme that we have to impose some extra conditions on partial solutions and there is usually a choice which conditions we want to impose.

Many problems in combinatorics, computational geometry, and computer science ask for the realization or existence of a certain object, for instance: a triangulation, an independent set, a plane drawing of a given graph, a plane drawing of a graph with prescribed areas, or a set of points in the plane realizing a certain combinatorial structure.
These questions become extension questions as soon as a part of the solution is already created and the question is, whether there is a way to finish it.
For instance: can a given plane matching be extended to a perfect matching, can a given plane graph be extended to a $3$-regular graph. When going to an extension question the difficulty might increase. This is  the case for matching extensions: Given a set with an even number of points in the plane there always exists a perfect, straight-line, crossing-free matching on that set of points. However, it is \NP-hard to decide if a given partial matching can be extended to a perfect one~\cite{DBLP:journals/corr/abs-1206-6360}. This is a good reason why our  results must be considered with care as~\cref{thm:hyper} is by its nature an extension result.

We take this one step further and ask whether there is a way to complete \emph{every} ``reasonable'' partial realization. As explained before, we might need to impose some problem-specific conditions, for the question to become meaningful. 

Without giving explicit definitions, we mention here a few more problems that are worth considering:
\begin{itemize} [topsep=1pt,itemsep=-1pt,partopsep=1pt,parsep=1pt,leftmargin=0.8cm]
\item Order-Type Extension~\cite{Mnev,matousekSegments},
\item Extension variant of the Steinitz problem~\cite{richter1995realization},
\item Graph Metric Extension~\cite{Muller2017},
\item Simultaneous Graph Embedding Extension~\cite{JGAA-218}.
\end{itemize}

\paragraph*{Imprecision.}
\newcommand{\imGuarding}{\textsc{Guarding under Imprecision}}
\newcommand{\unGuarding}{\textsc{Universal Guard Set}}

Again, we use the \AGP as an example. 
We introduce the \textsc{Imprecise Guarding} and \textsc{Universal Guarding} problems.
The underlying idea is to guard a polygon, but we know the polygon only in an imprecise way. One might think of two different scenarios. In the first scenario, we want to know whether it is always possible to guard the polygon with $k$ guards, no matter how the actual polygon behaves. In the second scenario, we want to find a set of \emph{universal guards} that will guard every possible polygon. \autoref{fig:UniGuarding} depicts an example of a universal guard set.

\begin{figure}[htbp]
 \centering
 \includegraphics[page=2]{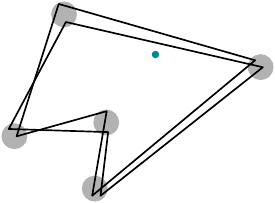}
 \caption{The cyan guard sees every point of the imprecisely given polygon. Therefore, it is a universal guard (set).}
 \label{fig:UniGuarding}
\end{figure}

\begin{framed}
\noindent {\bf \imGuarding}\\
\noindent {\bf Input:} A set of unit disks $d_1,\ldots,d_n$ and a number $k\in \N$. \\
\noindent {\bf Question:} Is it true that for every set of points $p_1\in d_1, \ldots ,p_n\in d_n$ the polygon described by $p_1,\ldots,p_n$ can be guarded by $k$ guards?
\end{framed}

\begin{framed}
\noindent {\bf \unGuarding}\\
\noindent {\bf Input:} A set of unit disks $d_1,\ldots,d_n$ and a number $k\in \N$. \\
\noindent {\bf Question:} Is there a set $G$ of $k$ guards such that  for every set of points $p_1\in d_1, \ldots ,p_n\in d_n$ the polygon described by $p_1,\ldots,p_n$ is guarded by $G$?
\end{framed}

Both solution concepts seem sensible as they are able to deal with the situation that the input is not known precisely, nevertheless, they guarantee a precise solution. Algorithms solving this problem must be more fault-tolerant in the sense that they are forgiving towards small errors in the input. It is easy to see that \textsc{Guarding under Imprecision} is contained in \FER and that \textsc{Universal Guard Set} is contained in \EFR. 
Using the Cook-Levin analog (\cref{thm:CookLevin}) may be the simplest way to see this.
Therefore, we wonder:
\begin{q}
 Is \textsc{Guarding under Imprecision} \FER-complete?
\end{q}
\begin{q}
 Is \textsc{Universal Guard Set} \EFR-complete?
\end{q}
For any geometric problem, we can ask whether there is a universal solution that still works under any small perturbation. However, most problems with geometric input are contained in \NP. On the other hand, problems that are known to be \ER-hard have usually combinatorial input. The \AGP 
has geometric input and
is \ER-complete.
These features make it distinct 
from many other geometric problems.
This means there might not be too many natural problems, which become \FER-complete if the algorithm is required to deal with imprecision.
Another exception to this rule of thumb are problems related to linkages~\cite{Connelly2003}. 

Let us give in this context another example 
to clarify these observations. Consider 
a unit-disk intersection graph $G$ 
given explicitly by a set of disks in 
the plane. We can ask for a dominating 
set. Now the problem becomes more 
challenging, if we ask if the graph 
can be dominated even after some 
small perturbation of the disks.
The perturbations of the input can also 
be understood as the inherent imprecision of the input. 
\begin{framed}
\noindent {\bf \textsc{Domination of Imprecise Unit Disks}}\\
\noindent {\bf Input:} A set of unit disks 
$d_1,\ldots,d_n$, a number $k\in \N$, and 
a number $\delta\in \Q$. \\
\noindent {\bf Question:} Is it true
that for every translation of each disk 
by a vector of length at most $\delta$, 
it is possible to dominate 
the resulting disk intersection 
graph by $k$ disks? 
\end{framed}
Note that the problem of finding a dominating set 
in a graph is contained in \NP. However, it 
is unclear if this remains the case even with the
perturbation. Additionally, the perturbation  
might be captured best with real-valued variables.
However, as a dominating set can be described in 
a discrete way, it looks unlikely to be \FER-complete.
However the following weaker question might have a positive answer.
\begin{q}
 Is \textsc{Imprecise Domination of Unit Disks} $\Sigma^p_2$
 and $\ER$-complete?
\end{q}
Indeed, we think that it might be true that
 there is a plausible class between 
 $\Sigma^p_2$ and \FER capturing imprecision issues.

\paragraph*{Robustness.}
As above, we use the \AGP as an illustrative example and we define the \textsc{Robust Guarding} problem as follows. 
\begin{framed}
\noindent {\bf\textsc{Robust Guarding}}\\
\noindent {\bf Input:} A simple polygon $P$ and a number $k\in \N$. \\
\noindent {\bf Question:} Does there exist a set of unit disks $D= \{d_1,\ldots,d_k\}$, each fully contained inside the polygon, such that for every placement $G = \{g_1,\ldots,g_k\}$ of guards with $g_i\in d_i$, the placement $G$ guards the whole polygon $P$?
\end{framed}
\begin{figure}[htbp]
	\centering
	\includegraphics[page=3]{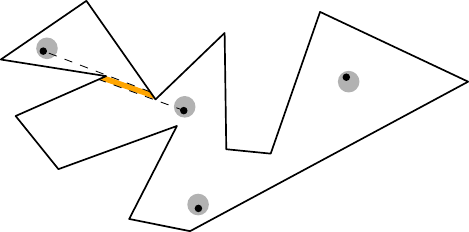}
	\caption{This guarding is not robust. The orange region might not be guarded.}
\end{figure}

It can be easily seen that this problem is contained in $\EFR$. The positions of the disks are existentially quantified, the positions of the guards are universally quantified, and the remaining formula enforces the guards to be inside the disks, the disks inside the polygon, and the condition that the guards are indeed guarding the whole polygon. 

\begin{q}
 Is \textsc{Robust Guarding} \EFR-complete?
\end{q}

In more general terms we employ the notion that a solution is robust, if it remains a valid solution also if it gets slightly perturbed. 
Another example is \textsc{Robust Order Type Realizability}. 
Given a set of points $P$ in the plane, each ordered triple is either in clockwise or counter clockwise direction. 
In \textsc{Order Type Realizability}, we are given an orientation for each triple of some abstract set and we are asked to find a set of points in the plane with that given order type. Robustness  asks if the solution remains correct, even after some small perturbation. As order types stay fixed under uniform scaling it makes sense to restrict the points to lie in a unit square.
\begin{framed}
\noindent {\bf\textsc{Robust Order Type Realizability}}\\
\noindent {\bf Input:} An abstract order type of $n$ points and a rational number $w>0$.\\
\noindent {\bf Question: } Does there exists a set of disks $d_1,\ldots,d_n$ with radius $w$ in the unit square $[0,1]^2$ such that every set of points $p_1\in d_1,\ldots , p_n\in d_n$ has the given order type?
\end{framed}
\begin{q}
 Is \textsc{Robust Order Type Realizability} \EFR-complete?
\end{q}

We assume here that $w$ is encoded in binary. Furthermore, if an order type is realizable for some $w_1$ then it is also robustly realizable for $w_2<w_1$. Thus this robustness problem becomes an optimization problem with respect to the parameter $w$.
It is conceivable that at least an approximate solution to the problem can be found by restricting the centers of the disks to a fine grid. 
This would imply that we can find a polynomial-sized witness for an approximate solution and thus the robust problem might become easier than the basic version for all these problems. Recall that for most of the mentioned recognition problems are \ER-complete~\cite{cardinal2017intersection,matousekSegments,mcdiarmid2013integer,schaefer2013realizability}.

In a similar spirit, we can define the robust variants for all kinds of recognition problems of intersections graphs,  such as intersection graph of  unit disks, segments, rays, unit segments, or of your favorite geometric object.
\begin{q}
 Is \textsc{Robust Recognition of Intersection graphs} \EFR-complete?
\end{q}

\begin{figure}[htbp]
 \centering
 \includegraphics[page=4]{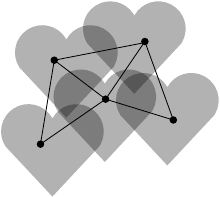}
 \caption{An intersection graph of hearts.}
 \label{fig:HeartShapes}
\end{figure}

\paragraph*{Environmental Nash-Equilibria.}
Finding Nash equilibria~\cite{GVK369342747} 
is one of the most important problems in game 
theory and it is known to be 
PPAD-complete~\cite{Daskalakis:2006:CCN:1132516.1132527}.
As every game has a Nash equilibrium, there 
is no point in formulating a corresponding 
decision problem.
However, if we restrict our attention to 
a small region of the strategy space, the 
problem becomes \ER-complete already for 
two players~\cite{DBLP:journals/mst/SchaeferS17}. 
Restricting the strategy spaces makes sense, 
as one might be looking for an exact solution ``close'' 
to a given approximate one.
A natural question we could ask is whether in 
a three-player game the first two players 
can find a Nash equilibrium for any fixed 
behavior of the third player. The underlying 
idea is that the third player simulates a 
potentially changing environment.

\subsection{A false candidate -- Universal Graph Metric}
Given a graph $G=(V,E)$ together with some 
edge weights $w:E\rightarrow \R^+$, we can 
ask for an embedding $\varphi$ of $G$ in the plane such 
that for each edge $u,v$ the distance 
$\textrm{dist}(u,v)$ in the plane equals 
the edge weight~$w(uv)$. 
In this case we say $\varphi$ realizes the edge-weight $w$.
Interestingly, it is  \ER-complete to decide if
an edge weight with all weights equal $1$
can be realized~\cite{schaefer2013realizability}.

Many edge weights, are trivially not
attainable, as they might not
satisfy the triangle inequality.
For the purpose of concreteness let us say that
a metric is \emph{reasonable} if there exists a
Euclidean Space of some dimension, into 
which the graph is embeddable.
It is easy to see that the dimension
can be upper bounded by the number of vertices.

We are ready to define the \textsc{Universal Graph Metric} problem. 

\begin{framed}
\noindent {\bf{\textsc Universal Graph Metric}}\\ 
   \noindent {\bf Input:} A graph $G$.\\   
    \noindent {\bf Question:} Is it true that for 
    every reasonable edge weight function, there exists 
    a realizing embedding in the Euclidean plane?
 \end{framed}

\begin{figure}[htbp]
	\centering
	\includegraphics[page=5]{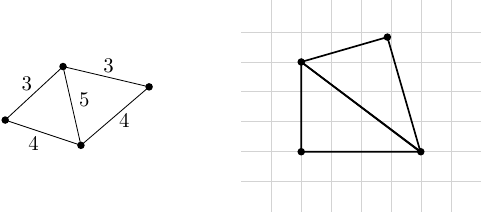}
	\caption{A graph with given edge weights 
	(on the left) and an embedding in the 
	Euclidean plane realizing the edge weights (on the right).}
\end{figure} 
 
\textsc{Universal Graph Metric} clearly has a 
similar spirit as area-universality. 
The edge lengths are universally quantified and 
the drawing is existentially quantified. 
However, a simple argument shows the 
following (surprising), but known, fact:

\begin{proposition}[Belk, Connelly~\cite{BelkC07}]
 \textsc{Universal Graph Metric} is in P.
\end{proposition}
\begin{proof}
The argument builds on the powerful theory 
of minor-closed graph classes. 
Recall that a graph $F$ is a \emph{minor} 
of a graph $G$ (denoted by $F\prec G$), if we can attain 
$F$ from $G$ by deleting edges, 
vertices, or contracting edges.
A graph class $\G$ is minor-closed, if $G \in \G$ 
and $F\prec G$ imply $F \in \G$.
A celebrated and deep theorem of Robertson and 
Seymour~\cite{ROBERTSON2004325} 
implies that every 
minor-closed graph class can be characterized 
by a finite number of forbidden minors. Since 
testing if a graph $G$ contains a fixed graph 
$F$ as a minor can be tested in polynomial time,  
the membership of a graph in any minor-closed 
graph class is checkable in polynomial time.

\begin{figure}[htbp]
	\centering
	\includegraphics{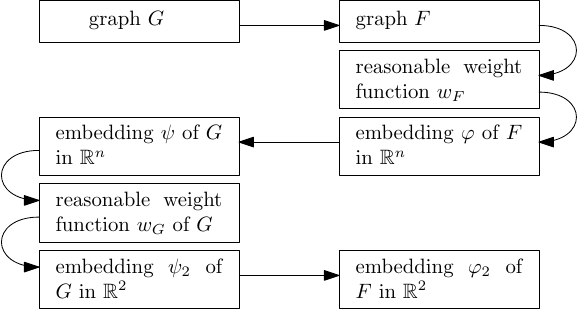}
	\caption{ \textsc{Universal Graph Metric} is
	a minor closed property.}
	\label{fig:MinorClosed}
\end{figure} 

It is easy to observe that the implicitly 
defined graph class of \textsc{Yes}-instances 
of the \textsc{Universal Graph Metric} problem is 
indeed minor-closed: Let $\mathcal G$ denote 
the class of \textsc{Yes}-instances for all \emph{reasonable} 
weight function. For the following argument refer
to~\cref{fig:MinorClosed}.
Let $G\in\mathcal G$ be a graph with a minor $F$.
Let $w_F$ be a reasonable weight function of $F$.
We need to show that $F$ has an embedding
respecting $w_F$ into $\R^2$.
(Here $n$ denote the number of vertices of $G$, we can assume
that the embedding is w.l.o.g. into $\R^n$ 
as we have at most $n$ vertices.)
For that purpose consider an embedding $\varphi$
of $F$ into $\R^n$ realizing $w_F$.
This exists by definition since $w_F$ is reasonable.
We can extend $\varphi$ to an embedding $\psi$
of $G$. For contracted edges, we use weight $0$,
deleted vertices are placed arbitrarily,
and weights of deleted edge  are inferred by the embedding.
This defines a weight function $w_G$
of $G$, which is, as we have just shown, reasonable.
Thus there is an embedding $\psi_2$ of 
$G$ into $\R^2$ as $G \in \mathcal{G}$.
This embedding restricted to the vertices of $F$
gives an embedding of $w_F$ into $\R^2$.
Thus $F\in \mathcal{G}$.
\end{proof}

It is easy to see that if we require just the triangle inequality, then $K_4$ is a forbidden minor
and it is known to be the only one~\cite{Belk07,BelkC07}.
If we choose \emph{reasonable} to be 
all weight functions, then due to the 
triangle inequality, the \textsc{Yes}-instances 
are exactly the set of trees and the 
forbidden minor is the triangle. 
In a recent Master Thesis by 
Muller~\cite{Muller2017} other meanings 
of \emph{reasonable} are discussed.

It is worth noting that the same proof does not work for 
\AU. It is easy to see that if $G$ is a graph and 
we remove either an edge or a vertex, then the resulting 
graph is still area-universal.
The same does not hold for edge contractions.
To see this, let~$G$ denote a graph that is not area-universal.
Due to Kleist~\cite{kleist1}, we know that the
$1$-subdivision~$F$ of $G$ is area-universal.
We contract the subdivided edges of $F$
one by one until we obtain $G$.
At some point in this process, the contraction
of an edge destroys the area-universality.

\subsection*{Acknowledgments}

We thank Jeff Erickson for some clarifications regarding
the Cook-Levin-type theorem.
%
%
%\subsection*{Data Availability Statement}
%Data sharing not applicable to this article as no datasets were generated or analysed during the current study.

%\newpage
\bibliographystyle{plainurl}       
\bibliography{Bib}

\end{document}